\documentclass[acmsmall,screen,nonacm]{acmart}
\pdfoutput=1

\usepackage[utf8]{inputenc}
\usepackage{xcolor}
\usepackage{subcaption}
\usepackage{thmtools}
\usepackage{thm-restate}

\usepackage{amsmath,amsthm,amsfonts,amscd,stmaryrd}
\usepackage[capitalise]{cleveref}
\usepackage{hhline}
\usepackage{multirow}
\usepackage{hyperref}
\usepackage{tikz}
\usetikzlibrary{shapes.geometric,shapes.misc,positioning,intersections,calc}
\usepackage{wrapfig}

\usepackage[noline,noend,ruled,linesnumbered]{algorithm2e}
\definecolor{light-gray}{gray}{0.9}

\SetInd {0.0em}{0.7em} 

\crefname{algocf}{Algorithm}{Algorithms}
\Crefname{algocf}{Algorithm}{Algorithms}

\theoremstyle{definition}
\newtheorem{definition}{Definition}[section]

\theoremstyle{definition}
\newtheorem{example}{Example}[section]

\newtheorem{theorem}{Theorem}[section]

\newtheorem{lemma}[theorem]{Lemma}

\newtheorem{problem}[theorem]{Problem}
\newtheorem*{remark}{Remark}

\newcommand{\dom}{\text{dom}}
\newcommand{\defeq}{\triangleq}
\newcommand{\tuple}[1]{\left\langle #1 \right\rangle}     
\newcommand{\set}[1]{\left\lbrace #1 \right\rbrace}       
\newcommand{\seq}[2]{\left\lbrace #1 \right\rbrace_{#2=0}^\infty}       

\newcommand{\I}[2]{\mathbf{I}_{#1}\left(#2\right)}

\newcommand{\LM}{\mathrm{LM}}

\newcommand{\imul}{\cdot} 
\newcommand{\istar}[1]{#1^*} 
\newcommand{\fstar}[1]{#1^\oast} 
\newcommand{\iexp}[2]{#1^{#2}}

\newcommand{\Det}[2]{\textsf{Det}\left(#1,#2\right)}

\newcommand{\TF}{\textbf{TF}}
\newcommand{\LIRR}{\textbf{LIRR}}
\newcommand{\dvext}[1]{\overline{#1}} 
\newcommand{\invimage}[2]{\dvext{#1}^{-1}\left[#2\right]} 
\newcommand{\image}[2]{\dvext{#1}\left[#2\right]} 

\newcommand{\iformula}[1]{\mathbf{F}(#1)} 

\newcommand{\gspace}[1]{\textit{span}\left(#1\right)}
\newcommand{\gideal}[1]{\left\langle#1\right\rangle}
\newcommand{\galg}[1]{\textit{alg}\left(#1\right)}

\usepackage{pifont}
\newcommand{\cmark}{\ding{51}}%
\newcommand{\xmark}{\ding{55}}%

\newcommand{\elimorder}[1]{\ll_{#1}}

\usepackage{xcolor}
\usepackage{listings}

\definecolor{mGreen}{rgb}{0,0.6,0}
\definecolor{mGray}{rgb}{0.5,0.5,0.5}
\definecolor{mPurple}{rgb}{0.58,0,0.82}
\definecolor{backgroundColour}{rgb}{1.0,1.0,1.0}
\lstdefinestyle{CStyle}{
    backgroundcolor=\color{backgroundColour},   
    commentstyle=\color{mGreen},
    keywordstyle=\color{magenta},
    numberstyle=\tiny\color{mGray},
    stringstyle=\color{mPurple},
    basicstyle=\footnotesize,
    breakatwhitespace=false,         
    breaklines=true,                 
    captionpos=b,                    
    keepspaces=true,                 
    numbers=left,                    
    numbersep=5pt,                  
    showspaces=false,                
    showstringspaces=false,
    showtabs=false,                  
    tabsize=2,
    language=C
}

\newcommand{\Tool}{{\textsc{Abstractionator}}\xspace}

\DeclareFontFamily{U}  {MnSymbolC}{}
\DeclareFontShape{U}{MnSymbolC}{m}{n}{
    <-6>  MnSymbolC5
   <6-7>  MnSymbolC6
   <7-8>  MnSymbolC7
   <8-9>  MnSymbolC8
   <9-10> MnSymbolC9
  <10-12> MnSymbolC10
  <12->   MnSymbolC12}{}
\DeclareSymbolFont{MnSyC}{U}{MnSymbolC}{m}{n}
\DeclareMathSymbol{\righthalfcup}{\mathrel}{MnSyC}{184}
\newcommand{\romanqed}{$\righthalfcup$}

\newcommand{\romanqedhere}{\let\qedsymbol\romanqed\qedhere}

\usepackage{booktabs}
\usepackage{soul}

\newif\iflong
\longtrue
\iflong
\newcommand{\refappendix}[1]{\sectref{#1}}
\usepackage[appendix=inline]{apxproof} 
\else
\newcommand{\refappendix}[1]{the supplementary materials}
\usepackage[appendix=strip]{apxproof} 
\fi

\newtheoremrep{theorem}{Theorem}[section]
\newtheoremrep{corollary}{Corollary}[theorem]
\newtheoremrep{lemma}[theorem]{Lemma}
\newtheoremrep{proposition}[theorem]{Proposition}

\iflong
\else
\setcopyright{rightsretained}
\acmDOI{10.1145/3632867}
\acmYear{2024}
\copyrightyear{2024}
\acmSubmissionID{popl24main-p122-p}
\acmJournal{PACMPL}
\acmVolume{8}
\acmNumber{POPL}
\acmArticle{25}
\acmMonth{1}
\received{2023-07-11}
\received[accepted]{2023-11-07}
\fi

\begin{document}

\begin{abstract}
  This paper presents a program analysis method that generates program summaries involving polynomial arithmetic. Our approach builds on prior techniques that use solvable polynomial maps for summarizing loops. These techniques are able to generate \emph{all} polynomial invariants for a restricted class of programs, but cannot be applied to programs outside of this class---for instance, programs with nested loops, conditional branching, unstructured control flow, etc.
There currently lacks approaches to apply these prior methods to the case of general programs. This paper bridges that gap. Instead of restricting the kinds of programs we can handle, our method \emph{abstracts} every loop into a model that can be solved with prior techniques, bringing to bear prior work on solvable polynomial maps to general programs. While no method can generate all polynomial invariants for arbitrary programs, our method establishes its merit through a \emph{monotonicty} result. We have implemented our techniques, and tested them on a suite of benchmarks from the literature. Our experiments indicate our techniques show promise on challenging verification tasks requiring non-linear reasoning.
\end{abstract}

\title{Solvable Polynomial Ideals: The Ideal Reflection for Program Analysis}

\author{John Cyphert}
\orcid{0009-0007-6310-413X}
\affiliation{%
  \institution{University of Wisconsin-Madison}
  \city{Madison}
  \state{WI}
  \country{USA}
}
\email{jcyphert@wisc.edu}

\author{Zachary Kincaid}
\orcid{0000-0002-7294-9165}
\affiliation{%
  \institution{Princeton University}
  \city{Princeton}
  \state{NJ}
  \country{USA}
}
\email{zkincaid@cs.princeton.edu}



\maketitle

\section{Introduction} \label{sec:introduction}
There has been a long history of prior work that automatically generates polynomial invariants. One line of work in this direction seeks to generate \emph{all} possible polynomial invariants for a restricted class of programs \cite{ISAAC:RCK2004,TACAS:Kovacs2008,VMCAI:HJK2018,LICS:HOPW2018,JACM:EOPW2023}. These complete methods give strong, predictable results; however, there is no obvious way to use such techniques for general programs, which may contain nested loops, branching, and unstructured control flow. Another line of research into the automatic generation of polynomial invariants looks to apply to general programs; however, such techniques are often heuristic in nature \cite{FMCAD:FK2015,POPL:KCBR2018} or are limited on what kind of invariants they can produce \cite{POPL:SSM2004,SAS:CJJK2012,ATVA:OBP16,MULLEROLM2004}, e.g. returning only polynomials up to some degree.

It is impossible to fully bridge the gap between these two lines of research. No method can generate all polynomial invariants for general programs. No method can even generate all \emph{linear} invariants for general programs \cite{MOS2004}.
However, the two lines of research raises the question: Can a method that generates polynomial invariants and works for general programs provide \emph{some} guarantee on predictability?

In this paper we present techniques to give a positive answer to the previous question.
Our method builds on the \emph{algebraic program analysis} framework \cite{CAV:KRC2021}. Within this framework, summaries are created for larger and larger subprograms in a bottom-up manner. The essential challenge is the summarization of loops. In short, summarizing a loop amounts to over-approximating the reflexive transitive closure of a transition formula that describes the loop body. Once an appropriate loop summarization technique is constructed, the algebraic framework can then employ the technique as a subroutine in the analysis of whole programs.

The technique for summarizing loops described in this paper works by \emph{abstracting} a transition formula describing an arbitrary loop body
to an object, which we call a \emph{transition ideal}. Informally, a transition ideal is a set of polynomial equations describing the transition relation of a loop body. Checking whether a non-linear transition formula (with an integrality predicate) implies a polynomial equation is undecidable for the standard model; however, \citet{POPL:KKZ2023} developed a theory, \LIRR,~for which it is possible to compute \textit{all} implied polynomial equations. 
Our work utilizes \LIRR~to extract transition ideals from loop body transition formulas. The extraction of transition ideals is complete for \LIRR~and sound for the standard interpretation.
A transition ideal can be generated as a summary of an inner loop, a summary of a program with branches, etc. To summarize the loop, we would like to compute the transitive closure of the extracted transition ideal; however, the dynamics of a transition ideal can be chaotic and difficult to capture. Thus, our key insight is that we need to again abstract the transition ideal to some other object for which we know how to compute invariants. In \cref{sec:reflection} we show how given an arbitrary transition ideal one can compute a \emph{best abstraction} as a \emph{solvable transition ideal}, which we call its \emph{solvable reflection}.
A solvable transition ideal is a transition ideal that contains all of the defining polynomials of at least one \emph{solvable polynomial map}, a class of polynomial maps that have been utilized in prior work on complete polynomial invariant generation \cite{TACAS:Kovacs2008,VMCAI:HJK2018,ISAAC:RCK2004,SAS:ABKKMS2022}. 
In \cref{sec:closure} we show that the method of \citet{KAUERS2008} can be generalized to compute the polynomial invariants of a solvable transition ideal. These resulting polynomial invariants can be translated back to a transition formula, which gives a method for summarizing arbitrary loops. Hence, via the algebraic framework we achieve a method to generate polynomial invariants to programs with arbitrary control-flow.

While our method is not complete for arbitrary programs, we can guarantee our method is \emph{monotone}. The exact definition of monotonicity is given in \cref{sec:summary}; informally though, a program analysis is monotone if ``more information in yields more information out''. That is, improving the precision of a code-fragment, e.g. by strengthening the precondition or adding assumptions, necessarily improves the overall analysis result. 
Our method is monotone because we do not extract just \emph{some} solvable transition ideal from a loop body, but our method extracts the \emph{best} solvable transition ideal. Another way to understand this result is that while our method is not complete for general programs, it is complete, in a sense, at every loop. Given a summary for a loop body, we will always compute the solvable transition ideal that most closely approximates it. Thus, in the restricted case of a simple loop whose body is described by a solvable polynomial map, our method is complete; and, in general our method is monotone.

Summarizing, this paper presents a program analyzer that (1) produces \emph{non-linear summaries}, (2) works for polynomial programs with \emph{arbitrary control-flow}, and (3) is \emph{monotone}. We have implemented our method and our experiments show that it performs comparably to the top performers on the \texttt{c/ReachSafety-Loops} subcategory of the Software Verification Competition.

\iflong
\else
\pagebreak
\fi
\noindent Our paper makes the following contributions:
\begin{enumerate}
    \item We introduce transition ideals and solvable transition ideals. We further generalize solvable transition ideals with \emph{ultimately} solvable transition ideals.
    \item We show that transition ideals admit (ultimately) solvable \textit{reflections}.
        \begin{itemize}
        \item We present algorithms for computing solvable linear reflections (\cref{sec:reflection-computation}) and ultimately solvable reflections (\cref{sec:ultimately-solvable}) of transition ideals.  Linear reflections correspond to \textit{best abstractions} with respect to linear simulations.
            \item We generalize this method to compute best abstractions with respect to  \emph{polynomial simulations} of bounded degree (\cref{sec:polynomial-simulations}).
        \end{itemize}
    \item We present a complete algorithm for computing \emph{all} the polynomial invariants of (ultimately) solvable transition ideals.
    \begin{itemize}
        \item Our summarization algorithm utilizes a sub-algorithm that might be of independent interest. Our sub-algorithm generalizes the technique of \citet{KAUERS2008}, which computes the algebraic relations of c-finite sequences of rational numbers, to compute algebraic relations of c-finite sequences over an arbitrary $\mathbb{Q}$-algebra (\cref{prob:solvePolyMap}).
    \end{itemize}
    \item An implementation of the combination of the abstraction and transitive closure results yields a \textit{monotone} program analysis that produces polynomial invariants for polynomial programs.
\end{enumerate}

The rest of the paper is organized as follows. \Cref{sec:overview} illustrates the main features of our method on a challenging example. \Cref{sec:background} gives background on commutative algebra, polynomial ideals, and solvable polynomial maps. \Cref{sec:reflection} describes the method of extracting (ultimately) solvable transition ideals from arbitrary transition ideals. \Cref{sec:closure} describes the method of summarizing (ultimately) solvable transition ideals. \Cref{sec:summary} connects these ideas to transition formulas, and shows how the methods can be integrated into a program analyzer. \Cref{sec:evaluation} presents the experimental evaluation. \Cref{sec:related} discusses related work.
\iflong
\else
Omitted proofs can be found in extended version of this document \cite{POPL:CK2024}.
\fi


\section{Overview}
\label{sec:overview}
\begin{figure}[t]
\begin{subfigure}[t]{0.45\textwidth}
    \centering
    \begin{lstlisting}[style=CStyle]
    a = x; b = y; p = 1; q = 0;
    r = 0; s = 1; c = 0; k = 0;
    while (b != 0) {
      c = a; k = 0;
      while (c >= b) {
        c = c - b;
        k = k + 1;
      }
      a = b;
      b = c;
      p, q = q, p - q * k;
      r, s = s, r - s * k;
    }
    assert(q*x + s*y == 0);
    assert(p*x + r*y == a);
    \end{lstlisting}
    \caption{This program implements the extended Euclidean algorithm. The program is a modified version of the verification task egcd2-ll.c (\url{https://github.com/sosy-lab/sv-benchmarks/blob/master/c/nla-digbench/egcd2-ll.c}). The assignments on lines 11 and 12 are parallel assignments.}
    \label{fig:overviewEx}
    \end{subfigure}
    \begin{subfigure}[T]{0.45\textwidth}
    \begin{tikzpicture}
\node (form) [label=left:{Loop body formula}] {$F(X, X')$};

\node (ideal) [below = 0.75cm of form, label = left:{Transition ideal}] {$T \in \mathbb{Q}[X, X']$};

\node (solv) [below = 0.75cm of ideal, label = left:{Solvable ideal}] {$U \in \mathbb{Q}[Y,Y']$};

\node (sol) [below = 0.75cm of solv, label = left:{Ideal summary}] {$U^* \in \mathbb{Q}[Y,Y']$};

\node (formsol) [below = 0.75cm of sol, label=left:{Loop body summary}] {$\fstar{F(X, X')}$};

\draw[thick,->,>=stealth] (form) -- node[inner sep=10pt, anchor=west, label=left:{(\Cref{sec:summary})}] {Step 1\label{step1}} (ideal);
\draw[thick,->,>=stealth] (ideal) --  node[inner sep=10pt, anchor=west, label=left:{(\Cref{sec:reflection})}] {Step 2\label{step2}} (solv);
\draw[thick,->,>=stealth] (solv) -- node[inner sep=10pt, anchor=west, label=left:{(\Cref{sec:closure})}] {Step 3\label{step3}} (sol);
\draw[thick,->,>=stealth] (sol) -- node[inner sep=10pt, anchor=west] {Step 4\label{step4}} (formsol);
\end{tikzpicture}
\vspace{1cm}
\caption{Overview of the method}
\label{fig:overviewDiag}
    \end{subfigure}
    \caption{}
\end{figure}
In this section, we present our technique for program verification on the motivating example found in \cref{fig:overviewEx}. Two relevant features of this verification task is that (1) the program has a \emph{nested loop}; and (2) to verify the assertions at the end of the program, an invariant involving \emph{non-linear arithmetic} is required. This combination of nested loops and non-linear arithmetic presents a significant challenge for existing methods. 

Our approach to analyzing programs builds on the algebraic program analysis framework \cite{CAV:KRC2021}. Within this framework, analysis proceeds by producing \emph{transition formulas} for each program substructure. A transition formula, $F(X, X')$, is a formula over the program variables $X$ as well as their primed counterparts $X'$. Such a formula represents a relation over program states, where the unprimed variables correspond to the pre-state and the primed variables correspond to the post-state. Summaries for the sequencing and branching of program substructures corresponds to the transition formulas operations $F(X, X')\circ G(X, X')\defeq \exists X''. F(X, X'')\land G(X'', X')$ and $F(X, X')\oplus G(X, X')\defeq F(X, X')\lor G(X, X')$ respectively.
From these two operations, one can accurately summarize non-looping code. For example, a transition formula, $F_{i}$ for the inner-loop of \cref{fig:overviewEx} would look like $F_{i} \defeq c < b \land c' = c - b \land k' = k + 1 \land b' = b$.
Of course, we are interested in analyzing programs that do have loops, so an algebraic analysis for looping code must have an iteration operator $\fstar{F}$. The benefit is that once an iteration operator is created, the analysis can work for any loop, regardless of the underlying program structure.

\Cref{fig:overviewDiag} gives an overview of our iteration operator. We illustrate the method by discussing the analysis of the inner-loop of \cref{fig:overviewEx}. 
The goal of Step~\hyperref[step1]{1} (discussed in \cref{sec:summary}) is to extract a \emph{transition ideal} from the loop body. Informally, transition ideal $T$ corresponds to a transition formula that can be expressed as a conjunction of polynomial equations. That is, the transition ideal of a transition formula $F(X,X')$ is the set $\mathbf{I}(F) = \set{ p \in \mathbb{Q}[X,X'] : F \models p = 0}$ of polynomials that vanish on all models of $F$.  For example, for the transition formula for the inner-loop body, we have $T_i = \mathbf{I}(F_i) = \mathbf{I}(c < b \land c' = c - b \land k' = k + 1\land b' = b) = \gideal{c' - c + b, k' - k - 1, b'-b}$. 

The objective of Step~\hyperref[step2]{2} (presented in \cref{sec:reflection}) of our method is
to extract a \emph{solvable transition ideal} from $T_i$.
We say that a transition ideal $U$ is solvable if it contains a solvable polynomial map $p$ (a homomorphism $p : \mathbb{Q}[X] \rightarrow \mathbb{Q}[X]$ of a particular form, defined in \cref{sec:c-finite}), in the sense that $x' - p(x)$ belongs to $U$ for all variables $x$. Such a $p$ is called a \textit{solvability witness} for $U$.
For the transition ideal $T_i$, the extraction step is trivial because $T_i$ itself is solvable: the function $p_i$ mapping $\set{c \mapsto c-b, k \mapsto k+1, b \mapsto b }$ (which is affine, a special case of solvable) is a witness. Therefore, the result of the second step of our method is the solvable transition ideal $U_i = T_i$.

The task of Step~\hyperref[step3]{3} (presented in \cref{sec:closure}) of our method is to ``summarize'' the ideal $U_i$, as the transition ideal $U_i^* = \bigcap_{n=0}^\infty U_i^n$. Thinking of $U_i^n$ as a set of polynomial constraints that hold after $n$ iterations of $U_i$, $U_i^*$ represents the constraint that hold after \emph{any} number of iterations. 
The process of computing $U^*$ from a solvable transition ideal $U$ is the subject of \cref{sec:closure}. The basic idea is that we can ``solve'' the solvable witness $p$ by deriving a closed-form $p^n(x)$ for each $x \in X$. In the case of the inner-loop of \cref{fig:overviewEx}, we have $p_i^n(c) = c - b n, p_i^n(k) = k + n, p_i^n(b) = b$. 
This solution represents the value of the program variables $c$, $k$, and $b$ after $n$ iterations of the loop. We can then obtain polynomial invariants by eliminating $n$. For our running example, we have $U_i^* = \gideal{c' - c + b(k' -k), b' - b}$.
Since $U_i$ happens to coincide with $\tuple{c - p_i(c), k - p_i(k), b - p_i(b)}$, computing
$U_i^*$ is essentially the same process as \citet{TACAS:Kovacs2008,KAUERS2008}'s complete invariant generation for solvable polynomial maps; \cref{sec:closure} shows how these ideas can be extended to solvable transition ideals in general.

The final step of our iteration operator (Step~\hyperref[step4]{4}) is to translate $U_i^*$ back to a transition formula, $\fstar{F(X, X')}$. For the inner-loop of \cref{fig:overviewEx}, the transition ideal $U_i^*$ translates to the transition formula $\fstar{F_i} \defeq c' - c - b(k'-k) = 0 \land b' - b = 0$. $\fstar{F_i}$ is our summary for the inner loop.

The example analysis of the inner loop of \cref{fig:overviewEx} gives the basic outline of how our method analyzes a loop. However, because the body of the inner loop implements a solvable polynomial map, Step 2 of \cref{fig:overviewDiag} was trivial. To understand the general case when a loop's body is \emph{not} described by a solvable polynomial map, consider the outer loop of \cref{fig:overviewDiag}. 
Let $F_o$ be a transition formula describing the outer-loop body, and let $T_o$ be a transition ideal obtained from $F_o$. $T_o$ contains many polynomials that do \emph{not} represent solvable assignments. For example, because the result of analysis of the inner loop $\fstar{F_i}$ is an approximation of the inner loop, the variable $k$ is updated non-deterministically in $T_o$; i.e., there is no $k' - p \in T_o$ for any polynomial $p$. Furthermore, $q$ has a non-linear self-dependence, i.e. $q' - p +qk' \in T_o$, which is not solvable. These complications mean that we cannot capture the dynamics of the variables of the outer loop using solvable polynomial maps.

However, we can find some \emph{terms} that evolve predictably. For example, $b$ and $c$ are always equal in the post-state, i.e. $b' - c' \in T_o$, and the sign of $qr - ps$ flips between the pre-state and post-state, i.e. $(q'r' - p's') + (qr - ps)\in T_o$. The evolution of these terms can be represented with a solvable transition ideal. \Cref{fig:outerloop-U} illustrates how the evolution of these terms for a single loop iteration can be represented by the solvable transition ideal $U'_o$. 
\begin{figure}[h]
    \begin{subfigure}{0.45\textwidth}
        \[
        \tuple{
        \left\{\begin{array}{rl}
        d &\mapsto qr - ps\\
        e &\mapsto b - c\\
        \end{array}\right\},
        \left(\begin{array}{l}
        d' + d\\
        e'\\
            \end{array}\right)}
        \]
        \caption{A solvable abstraction, $\tuple{u,U_o'}$, of $T_o$ \label{fig:outerloop-U}}
    \end{subfigure}
        \begin{subfigure}{0.45\textwidth}
        \[
                \left(\begin{array}{l}
        (q'r' - p's') + (qr - ps)\\
        (b' - c')
            \end{array}\right)
        \]
        \caption{Image of $U_o'$ under $\dvext{u}$ \label{fig:outerloop-image}}
    \end{subfigure}
    \caption{An abstraction of the outer-loop transition ideal \label{fig:outer-loop-sim}}
\end{figure}
Because $U'_o$ represents terms that are in $T_o$, the pair $\tuple{u, U'_o}$ \emph{abstracts} $T_o$. Let $Y$ be the set of variables $\set{d, e}$ and recall $X$ is used for the set of program variables. The variable $d$ in $U_o'$ represents the polynomial term $qr-ps$ and the variable $e$ represents the linear term $b-c$. This connection between the variables of $U_o'$ and the terms of $T_o$ is captured by the \emph{simulation} $u:\mathbb{Q}[Y]\rightarrow \mathbb{Q}[X]$, the polynomial homomorphism defined by $u(d) = qr-ps$ and $u(e) = b-c$. 
$\tuple{u, U_o'}$ is a sound abstraction of $T_o$ in the sense that
$\dvext{u}[U_o']$ is contained in $T_o$, where $\dvext{u}[U_o']$ denotes the image of $U_o'$ under the homomorphism $u$ extended to ``primed'' vocabulary by defining $\overline{u}(d) = u(d)'$ and $\overline{u}(e) = u(e)'$.

While $\tuple{u, U_o'}$ is \emph{an} abstraction of $T_o$, there could be other abstractions of $T_o$ that are better. For example, there are other polynomial terms that behave predictably that are not in $\dvext{u}[U_o']$, e.g. $(a's' - c'r') + (as - cr) - (br - cr)\in T_o$. Other abstractions may consider terms that are not captured by $\tuple{u, U_o'}$.
However, the techniques of \cref{sec:reflection} does not just extract \emph{a} sound abstraction, but actually extracts a \emph{best abstraction}, with respect to a class of simulations. We call such a best abstraction a \emph{solvable reflection}\footnote{The name solvable reflection is derived from the notion of a reflective subcategory in category theory.}. Informally, a solvable reflection is best in that any other abstraction also abstracts the solvable reflection. In \cref{sec:reflection-computation} we give an algorithm for producing a solvable reflection with respect to linear simulations. 
For the case of linear simulations, $\tuple{v, V}$ is a solvable reflection, with $V = \gideal{e'}\subseteq\mathbb{Q}[e, e']$, $v(e) = b-c$. In other words, capturing the dynamics of the linear term $b - c$ is the best among all possible abstractions of linear terms with solvable transition ideals. 

In \cref{sec:polynomial-simulations}, we extend our algorithm for finding linear simulations and give a method for producing solvable reflections with respect to polynomial simulations of a bounded degree. 
The simulation $u$ from \cref{fig:outer-loop-sim} is an example of a degree-2 simulation, i.e. the mapping for the variable $d$ is a degree-2 polynomial. Our extended method is able to produce the solvable reflection, $\tuple{t, U_o}$, with respect to degree-2 simulations, of $T_o$. $\tuple{t, U_o}$ is too big to be presented here; however, it necessarily captures more dynamics of the outer loop compared to $\tuple{u, U_o'}$. Furthermore, for this example, the closure, $U_o^*$, when combined with the program's initial conditions is strong enough to prove the two assertions at the end of the program, verifying the program in \cref{fig:overviewEx}.

The key that makes the overall process \emph{monotone} is the combination of best abstractions with complete invariant generation for solvable transition ideals. In other words, at every loop we are finding the strongest loop-body invariant that we know how to completely solve. This leads to the result that our iteration operator is monotone (\cref{sec:summary}). Moreover, in the case when the loop body is described by a solvable polynomial map, similar to the case of the analysis of the inner-loop, our method essentially reduces to prior methods. Consequently our method is complete in such a case.

\section{Background} \label{sec:background}

\subsection{Polynomials, Ideals, and Gr\"obner Bases}
We use $\mathbb{Q}[z_1, \dots, z_n]$ and $\mathbb{Q}[Z]$ to denote the ring of polynomials with rational coefficients over the variables $\set{z_1, \dots, z_n} = Z$.
A \textbf{polynomial homomorphism} is a ring homomorphism $f : \mathbb{Q}[X] \rightarrow \mathbb{Q}[Y]$ between two polynomial rings.  Provided that $X$ is finite, a polynomial homomorphism can be represented by its action on the variables $X$.
We say that $f$ is \textbf{linear} if for each $x \in X$, $f(x)$ is either 0 or a homogeneous polynomial of degree 1.  We say that $f$ is a \textbf{polynomial endomorphism} if $X = Y$. In this paper, every polynomial ring we consider is over a finite set of variables.

Next, we highlight standard definitions for polynomial ideals. For a more in depth presentation of these topics, \citet{Book:CLO2015} provides a good introduction.
A polynomial \textbf{ideal} $I\subseteq \mathbb{Q}[Z]$ is a set that contains $0$, is closed under addition, and for any $p\in I$ and $q\in \mathbb{Q}[Z]$, $pq\in I$. Intuitively, one can consider an ideal $I$ a collection of polynomial equations $\set{p=0 : p\in I}$. The conditions of an ideal can be read as inference rules: $0 = 0$, if $p=0$ and $q = 0$ then $p+q = 0$, and if $p = 0$ then $pq = 0$. For any collection of polynomials $P\subseteq\mathbb{Q}[Z]$, we use $\gideal{P} \defeq \set{g_1p_1 + \dots + g_lp_l : p_i\in P, g_i\in \mathbb{Q}[Z]}$ to denote the \textbf{ideal generated by} $P$.

A \textbf{monomial} $m$ is a product of variables of the form $m = z_1^{d_1}\dots z_n^{d_n}$. The \textbf{total degree} of $m$ is $d_1+\dots +d_n$. A \textbf{monomial order}, $\ll$, is a total ordering on monomials, such that for any monomial $v$, $1 \ll v$ and if $m \ll n$ then $mv \ll nv$. The \textbf{leading monomial}, $\LM(p)$, with respect to a given monomial order, of a polynomial $p = a_1m_1+\dots + a_nm_n$ is the greatest monomial among $m_1, \dots, m_n$. 
In this paper, we make use of two different types of monomial orders: \textit{graded orders} and \textit{elimination orders}. Graded orders first compare monomials by total degree, with larger degree corresponding to a larger monomial; ties in total degree are broken by some other monomial order. For example, a graded order that breaks ties using a lexicographic ordering on monomials is the \textbf{graded lexicographic order}. 
Let $X \cup Y$ be a partition of the variables $Z$. Let $m = m_xm_y$ and $n = n_xn_y$ be monomials with $m_x$ and $n_x$ containing only $X$ variables, and $m_y$ and $n_y$ only containing $Y$ variables. Let $\ll$ be some monomial order. 
The \textbf{elimination order} $\elimorder{X}$ defines $m\elimorder{X} n$ as either (1) $m_x \ll n_x$ or (2) $m_x = n_x$ and $m_y \ll n_y$.

\begin{example}
    Consider monomials over the variables $x$, $y$, and $z$ with $x$ lexicographically greater than $y$ and $y$ lexicographically greater than $z$. Let $\ll_{\text{grlex}}$ be the graded lexicographic order, and let $\elimorder{\{z\}}$ be the elimination order that eliminates $z$ and uses $\ll_{\text{grlex}}$ for remaining comparisons.
    \begin{itemize}
        \item $z^2\ll_{\text{grlex}} x^2 \ll_{\text{grlex}} x^2z^2 \ll_{\text{grlex}} xy^2z$
        \item $x^2\elimorder{\{z\}} xy^2z \elimorder{\{z\}} z^2 \elimorder{\{z\}} x^2z^2$
    \end{itemize}
\end{example}

Fixing a monomial ordering, $\ll$, every polynomial ideal $I\subseteq\mathbb{Q}[X]$ admits a finite \textbf{Gr\"obner basis} $G\subseteq\mathbb{Q}[X]$, with $\gideal{G} = I$. The exact definition of a Gr\"obner basis is unimportant for this paper. Instead we note the relevant properties of Gr\"obner bases that we need. Given a Gr\"obner basis, $G =\set{g_1, \dots, g_k}$ (w.r.t $\ll$) for an ideal $I$, every polynomial $p\in \mathbb{Q}[X]$ can be written with respect to $G$ as $p = c_1g_1+\dots+c_kg_k + r$, with $c_1, \dots, c_k, r\in \mathbb{Q}[X]$ such that:
\begin{itemize}
    \item $r$ is the unique polynomial with: (1) no term of $r$ is divisible by any $\LM(g_1), \dots, \LM(g_k)$, and (2) there is a $g\in I$ such that $p = g+r$. Consequently, $r$ has the property that $\LM(r)\ll \LM(r')$ for any $r'\in \mathbb{Q}[X]$ and $g' \in I$ with $p = g' + r'$.
    \item $\LM(c_ig_i)\ll \LM(p)$ for $i = 1, \dots, k$ and $\LM(r)\ll \LM(p)$.
\end{itemize}
Given a finite set of polynomials $P$ there are algorithms \cite{SIGSAM:Buchberger1976,FAUGERE199961} for computing a Gr\"obner basis of $\gideal{P}$. Furthermore, given a Gr\"obner basis, $G$, there are algorithms for rewriting a polynomial by $G$. 

\begin{example}
    Consider the graded lexicographic order over the variables $x$, $y$, and $z$. We will not go through the steps to calculate it but $\{z^2 - 1, x - 2, y + z\}$ is a Gr\"obner basis for the ideal $\gideal{y^2 - 1, yz + 1, xz^2 - 2}$. Thus, $\gideal{z^2 - 1, x - 2, y + z}=\gideal{y^2 - 1, yz + 1, xz^2 - 2}$.

    $xz^2 + x^2 - y$ can be written with respect to the Gr\"obner basis $\gideal{y^2 - 1, yz + 1, xz^2 - 2}$ as
    \[
    xz^2 + x^2 - y = x\underbrace{(z^2 - 1)}_{g_1} + (x+3)\underbrace{(x-2)}_{g_2} - \underbrace{(y+z)}_{g_3} + \underbrace{z+6}_r.
    \]
\end{example}

The combination of Gr\"obner bases and elimination orderings result in the \textbf{key property of elimination orderings}: Let $X$ and $Y$ be disjoint sets of variables, and let $G \subseteq \mathbb{Q}[X,Y]$ be a Gr\"{o}bner basis for $\gideal{G}$ w.r.t. $\elimorder{X}$. Then
$\gideal{G \cap \mathbb{Q}[Y]} = \gideal{G} \cap \mathbb{Q}[Y]$. This key property is critical to many of our algorithms and arguments.

\begin{example}
    With respect to the graded lexicographic order over the variables $x$ and $y$, $G = \{x^3 - y^2 - 1, y^3 - 4x^2 + y, xy - 4\}$ is a Gr\"obner basis for the ideal $\gideal{G}$. However, with respect to the elimination order that eliminates $x$, $\elimorder{\{x\}}$, $G' = \{16x - y^4 - y^2, y^5 + y^3 - 64\}$ is a Gr\"obner basis for the ideal $\gideal{G}$. By the key property of elimination orderings $\gideal{G} \cap \mathbb{Q}[y] = \gideal{G'} \cap \mathbb{Q}[y] = \gideal{G'\cap \mathbb{Q}[y]} = \gideal{y^5 + y^3 - 64}$. Informally, $\gideal{y^5 + y^3 - 64}$ is the ideal $\gideal{G}$ with the variable $x$ ``projected out''.
\end{example}

Let $X$ and $Y$ be finite set of variables, let $P \subseteq \mathbb{Q}[X]$ be a set of polynomials, and let $f : \mathbb{Q}[Y] \rightarrow \mathbb{Q}[X]$ be a polynomial homomorphism.  Then the inverse image 
$f^{-1}[\gideal{P}] \defeq \set{ q \in \mathbb{Q}[Y] : f(q) \in \gideal{P} }$ of $\gideal{P}$ under $f$ is an ideal of $\mathbb{Q}[Y]$, and it can be computed as follows.
Without loss of generality, we assume $X$ and $Y$ are disjoint.
Let $G$ be a Gr\"{o}bner basis for the ideal generated by
$P \cup \set{ y - f(y) : y \in Y }$, with respect to an elimination order $\elimorder{X}$. Define
 $\textit{inv.image}(f,P) \defeq G \cap \mathbb{Q}[Y]$.

 \begin{lemmarep}[Inverse image] \label{lem:inverse-image}
     Let $X$ and $Y$ be finite sets of variables, let $P \subseteq \mathbb{Q}[X]$ be a set of polynomials, and let $f : \mathbb{Q}[Y] \rightarrow \mathbb{Q}[X]$ be a polynomial homomorphism.  Then we have
     \[ \gideal{\textit{inv.image}(f,P)} = f^{-1}[\gideal{P}] \ . \]
 \end{lemmarep}
 \begin{appendixproof}
    Suppose that $G$ be a Gr\"{o}bner basis for the ideal generated by
$P \cup \set{ y - f(y) : y \in Y }$, so that
$\textit{inv.image}(f,P) = G \cap \mathbb{Q}[Y]$.
Observe that:
\begin{enumerate}
    \item[(O1)] For each $q \in \mathbb{Q}[Y]$, we have $q - f(q) \in \gideal{G}$ (by induction on $q$)
    \item[(O2)] $\gideal{G} \cap \mathbb{Q}[X] = \gideal{P}$: Without loss of generality, we may suppose that $P$ is a Gr\"{o}ber basis for $\gideal{P}$ w.r.t. $\elimorder{Y}$.  Then
    $P \cup \set{ y - f(y) : y \in Y }$ is also Gr\"{o}bner basis w.r.t. $\elimorder{Y}$, and so $\gideal{G} \cap \mathbb{Q}[X] = \gideal{G \cap \mathbb{Q}[X]} = \gideal{P}$ by the key property of elimination orderings.
\end{enumerate}

We show $\gideal{\textit{inv.image}(f,P)} = f^{-1}[\gideal{P}]$ by proving inclusion in both directions.
\begin{itemize}
\item[$\subseteq$]  It is sufficient to show that for all $q \in \gideal{\textit{inv.image}(f,P)}$, we have $f(q) \in \gideal{P}$. Suppose $q \in \gideal{\textit{inv.image}(f,P)}$. Then $q\in \gideal{G}$ and $q \in \mathbb{Q}[Y]$. Since $q\in \mathbb{Q}[Y]$,
  we have $q - f(q) \in \gideal{G}$ by observation (O1). Since $q - f(q) \in \gideal{G}$ and  $q\in \gideal{G}$, we have $q - (q - f(q)) = f(q) \in \gideal{G}$.  From $f(q) \in \gideal{G} \cap \mathbb{Q}[X]$, we may conclude $f(q) \in \gideal{P}$ by observation (O2).
\item[$\supseteq$] Suppose $q \in f^{-1}[\gideal{P}]$.   By observation (O1), $q - f(q) \in \gideal{G}$, and since $f(q) \in \gideal{P} \subseteq \gideal{G}$ by assumption, we have $q = q - f(q) + f(q) \in \gideal{G}$.  Since $G$ is a Gr\"{o}bner basis for $\gideal{G}$ w.r.t. $\elimorder{X}$, we have
$\gideal{G} \cap \mathbb{Q}[Y] = \gideal{G \cap \mathbb{Q}[Y]} = \gideal{\textit{inv.image}(f,P)}$ by the key property of elimination orderings. \qedhere
\end{itemize}
 \end{appendixproof}
 
\begin{example}
    Let $X = \{x, y\}$ and $Y = \{a, b\}$. Let $f : \mathbb{Q}[Y] \rightarrow \mathbb{Q}[X]$ be the polynomial homomorphism defined by $f(a) = x^2 + y^2$ and $f(b) = xy$, and let $P\subseteq\mathbb{Q}[X]$ be the set $\{x+y + 1\}$. A Gr\"obner basis for the ideal generated by $\{x + y + 1\} \cup \{a - (x^2 + y^2), b - xy\}$, with respect to an elimination order $\elimorder{X}$ is $G = \{a + 2b - 1, b + y^2 + y, x + y + 1\}$. $G\cap \mathbb{Q}[Y] = \{a + 2b - 1\}$. Thus, by \cref{lem:inverse-image}, $\gideal{a + 2b - 1} = f^{-1}[\gideal{x + y + 1}]$. We can easily verify one direction of the equality by observing that $f(a + 2b - 1) = x^2 + y^2 + 2xy - 1 = (x+y-1)(x + y + 1) \in \gideal{x + y + 1}$.
\end{example}

 Given two ideals $\gideal{I}=A\subseteq \mathbb{Q}[X]$ and $\gideal{J} = B\subseteq \mathbb{Q}[X]$, $A\cap B$ is also an ideal of $\mathbb{Q}[X]$. A basis for $A\cap B$ can be computed from $I$ and $J$ using Gr\"obner basis techniques. $A + B\subseteq\mathbb{Q}[X]$ is an ideal and represents the set $\set{a + b : a \in A, b\in B}$. $A + B$ is the smallest ideal containing $A$ and $B$, and $A+B = \gideal{I\cup J}$. For any ideal $I\subseteq\mathbb{Q}[X]$ and polynomial $p\in \mathbb{Q}[X]$, we denote the set $\set{q : p - q\in I}$ (equivalently, $\set{p+q : q\in I}$) as $p+I$. We use $\mathbb{Q}[X]/I$ to denote the ring with carrier $\set{p + I : p\in \mathbb{Q}[X]}$, with addition and multiplication lifted to sets.

\subsection{Commutative Algebra}
Define a \textbf{$\mathbb{Q}$-algebra} to be a commutative algebra over $\mathbb{Q}$; that is,
an algebraic structure that is both a commutative ring and a linear space over $\mathbb{Q}$.  Examples of $\mathbb{Q}$-algebras include $\mathbb{Q}$ itself, the field of algebraic numbers $\bar{\mathbb{Q}}$, $\mathbb{Q}[X]$, and
$\mathbb{Q}[X]/I$ for any set of variables $X$ and ideal $I \subseteq \mathbb{Q}[X]$.
For any set of variables $X$, a $\mathbb{Q}$-algebra $A$ defines an algebra homomorphism
$(-)^A : \mathbb{Q}[X] \rightarrow (A^X \rightarrow A)$, where $x^A(v) = v(x)$.\footnote{$\mathbb{Q}[X]$ is the free $\mathbb{Q}$-algebra generated by $X$; $(-)^A$ is the ``evaluation'' function that we get from freeness.}
For any set $S \subseteq A^X$, define the vanishing ideal of $S$ to
be \[\I{A}{S} \defeq \set{ p \in \mathbb{Q}[X] : p^A(v) = 0\ \text{for every}\ v\in S}\ .\]
Observe that for any polynomial endomorphism $f : \mathbb{Q}[X] \rightarrow \mathbb{Q}[X]$ and any $\mathbb{Q}$-algebra $A$, $f$ defines a function $f_A : A^X \rightarrow A^X$ by
  $f_A(v) = \lambda x. f(x)^A(v)$. For example, let $A = \mathbb{Q}[w]$, $X = \{x, y\}$, and $v = \{x \mapsto 2, y \mapsto w\} \in A^X$. Then let $f : \mathbb{Q}[x, y] \rightarrow \mathbb{Q}[x, y]$ be the polynomial homomorphism defined by $f(x) = 2x$ and $f(y) = y+1$. Then, $f_A(v) = \{x \mapsto 4, y \mapsto w + 1\}$.

Note that for any $\mathbb{Q}$-algebra $A$, the set of infinite sequences over $A$, $A^\omega$, is also a $\mathbb{Q}$-algebra. The multiplication and addition operations of $A^\omega$ are defined pointwise. Let $0_A$ and $1_A$ be the additive and multiplicative unit of $A$, then $\seq{0_A}{n}$ and $\seq{1_A}{n}$ are the additive and multiplicative unit of $A^\omega$. The scalar multiplication operation of $A^\omega$ is defined as applying the scalar multiplication of $A$ to each element of the infinite sequence.

For a $\mathbb{Q}$-algebra $A$, and a set $G \subseteq A$,
we use $\gspace{G}$ to denote the smallest subspace of $A$ that contains $G$, and $\galg{G}$ to denote the smallest $\mathbb{Q}$-subalgebra of $A$ that contains $G$.

\subsection{C-finite Recurrences}\label{sec:c-finite}
Let $A$ be a $\mathbb{Q}$-algebra. A sequence $\{a(n)\}_{n=0}^\infty \in A^\omega$ is \textbf{c-finite} if it satisfies a c-finite recurrence. A \textbf{c-finite recurrence} has the form:\footnote{Here we present c-finite recurrences in \emph{homogeneous} form. \emph{Inhomogeneous} c-finite recurrences can have an additional additive constant $c_{d+1}$ at the end of \cref{Eq:c-finite}. Any inhomogeneous c-finite recurrence can be transformed to a homogeneous c-finite recurrence of order 1 higher, and so no power is lost when only considering the homogenous form.}
\begin{equation}\label{Eq:c-finite}
    a(n) = c_1a(n-1) + \dots + c_da(n-d)
\end{equation}
for constants $c_i\in \mathbb{Q}$, for all $n \geq d$. Given a recurrence of the form from \cref{Eq:c-finite} the \textbf{order} of the recurrence is $d$. The \textbf{characteristic polynomial} of a c-finite recurrence of the form from \cref{Eq:c-finite} is $p(x) = x^d - c_1x^{d-1} - \dots -c_{d-1}x - c_d\in \mathbb{Q}[x]$. The Fibonacci sequence, $\seq{F(n)}{n}$, is a classical example of a c-finite sequence over the $\mathbb{Q}$-algebra $\mathbb{Q}$, which satisfies the order 2 recurrence $F(n) = F(n-1) + F(n-2)$. The characteristic polynomial of the Fibonacci recurrence is $p_{\text{Fib}}(x) = x^2 - x - 1$.

Every c-finite recurrence admits a closed-form as a polynomial-exponential \cite{everest2003recurrence}. More specifically, given a recurrence of the form from \cref{Eq:c-finite} with $c_i\in \mathbb{Q}$ and $d$ initial values $a(0), \dots, a(d-1)$ from some $\mathbb{Q}$-algebra $A$, then
\begin{equation}\label{Eq:c-finiteClosedForm}
   a(n) = \sum_{i=1}^d\left(z_i(n) + \sum_{j = 1}^d p_{ij}(n) \Theta_j^n\right)a(i-1)
\end{equation}
where each $\Theta_i$ is a complex root of the characteristic polynomial of the recurrence, each $p_{ij}\in \mathbb{\bar{Q}}[x]$\footnote{More specifically, each polynomial $p_{ij}$ has coefficients in the splitting field of the characteristic polynomial, $\mathbb{Q}(\Theta_1, \dots, \Theta_d)$}, and $z_i(n) : \mathbb{N} \rightarrow \mathbb{Q}$ with $z_i(k) = 0$ for any $k\ge d$. More specifically, $z_i(k) = 0$ for any $k$ greater than or equal to the multiplicity of 0 as a root of the characteristic polynomial. If $0$ is not a root of the characteristic polynomial then the terms $z_i(k) = 0$ for all $k\in \mathbb{N}$ and can be omitted from the closed-form. Determining such a closed-form from a recurrence is referred to as ``solving'' the recurrence. The roots of the Fibonacci characteristic polynomial $x^2 - x - 1$ are $\phi = \frac{1+\sqrt{5}}{2}$ and $\psi = \frac{1 - \sqrt{5}}{2}$. Assuming, $F(0) = 0$ and $F(1) = 1$, a solution to the Fibonacci recurrence in the form of \cref{Eq:c-finiteClosedForm} is Binet's formula $F(n) = \frac{1}{\sqrt{5}}\phi^n - \frac{1}{\sqrt{5}}\psi^n$.

A polynomial endomorphism $f : \mathbb{Q}[X] \rightarrow \mathbb{Q}[X]$ is
\textbf{solvable} if there exists a partition $X = X_1 \cup \dots \cup X_n$ of
$X$ (with $X_i \cap X_j = \emptyset$ for all $i \neq j$) such that for each
$X_i$ and each $x \in X_i$, $f(x)$ can be written as $g(X_i) + h(X_1,\dots,
X_{i-1})$, where $g$ is a linear polynomial in the variables $X_i$ and $h$ is
a polynomial (of arbitrary degree) in the variables $X_1 \cup \dots \cup
X_{i-1}$.

C-finite recurrences are equivalent to solvable polynomial maps in the sense that each solvable polynomial map $f:\mathbb{Q}[X]\rightarrow\mathbb{Q}[X]$ defines $|X|$ c-finite sequences $\{f^i(x_1)\}_{i=0}^\infty, \dots, \{f^i(x_n)\}_{i=0}^\infty$, each of order $|X|$ \cite[Section 8]{POPL:KCBR2018}. Conversely, each c-finite recurrence $a(n) = c_1a(n-1) + \dots + c_da(n-d)$ can be transformed to a solvable map $f$ over $d$ variables, $a_n, \dots, a_{n-d}$ as the homomorphism defined by $f(a_n) = c_1a_{n-1} + \dots + c_da_{n-d}$ and $f(a_{n-i}) = a_{n-i+1}$ for $0< i \le d$. Due to this equivalence, solvable polynomial maps can effectively be ``solved'' in the form of \cref{Eq:c-finiteClosedForm} in the same way as c-finite recurrences. Hence the name \emph{solvable} polynomial map.

\subsection{Transition Formulas and Linear Integer/Real Rings}

Fix a set of program variables $X$.  We use $X' = \set {x' : x \in X}$ to
denote a set of ``primed'' copies of variables in $X$ (presumed disjoint from
$X$).  We use $(-)'$ to denote the homomorphism $\mathbb{Q}[X] \rightarrow
\mathbb{Q}[X']$ that maps each $x$ to its primed copy $x'$.

A \textbf{transition formula} is a formula $F$ with free variables in $X$ and $X'$ (in some first-order language), with the unprimed variables representing the pre-state of some computation, and the primed variables representing the post-state.
For transition formulas $F_1$ and $F_2$, we use $F_1 \circ F_2 \defeq \exists X''. F_1[X' \mapsto X''] \land F_2[X \mapsto X'']$ to denote the sequential composition of $F_1$ and $F_2$.  For any transition formula $F$ and natural number $n$, we use $F^n \defeq F \circ \dots \circ F$ ($n$ times) to denote the $n$-fold sequential composition of $F$ with itself.

Within this paper, we shall assume that transition formulas are expressed
 in the existential fragment of the language of non-linear mixed integer/real arithmetic (that is, the language of rational constants, addition, multiplication, an order relation, and an integrality predicate).
Although this language is undecidable over the \textit{standard} model,
\citet{POPL:KKZ2023} showed that ground satisfiability is decidable if we allow more general interpretations, namely over linear integer/real rings (\textbf{LIRR}).  For our purposes, we may think of linear integer/real rings as $\mathbb{Q}$-algebras that satisfy some additional axioms concerning the order relation and integrality predicate, which are not relevant to this paper.  We will assume $\textbf{LIRR}$ as a background theory in the remainder of the paper, and use $F \models_{\LIRR} G$ to denote that the formula $F$ entails the formula $G$ modulo $\LIRR$.

In addition to satisfiability being decidable, there is a procedure \cite{POPL:KKZ2023} for computing
the vanishing ideal $\I{\LIRR}{F}$ of a formula $F$ in the existential fragment of the language: the ideal of all polynomials $p$ such that $F \models_{\LIRR} p = 0$. See \cref{Ex:idealFromForm} for an example of $\I{\LIRR}{F}$ from a formula $F$. For any ideal $I$ generated by polynomials $p_1,\dotsi,p_n$, we use $\iformula{I}$ to denote the formula $p_1 = 0 \land \dotsi \land p_n = 0$.  The choice of generators for $I$ is irrelevant in the sense that if two sets of polynomials $P$ and $Q$ generate the same ideal, then
$\bigwedge_{p \in P} p = 0$ and $\bigwedge_{q \in Q} q = 0$ are equivalent modulo $\LIRR$.  Note that $\mathbf{I}_{\LIRR}$ and $\mathbf{F}$ form a Galois connection: for any transition formula $F$ and ideal $I$ over the free variables of $F$, we have
$F \models_{\LIRR} \iformula{I}$ if and only if $\I{\LIRR}{F} \supseteq I$.  This implies that (1) for formulas $F$ and $G$, if $F \models_{\LIRR} G$, then $\I{\LIRR}{F} \supseteq \I{\LIRR}{G}$, and (2) for ideals $I$ and $J$ if $I \supseteq J$, then $\iformula{I} \models_{\LIRR} \iformula{J}$.

\subsection{Transition Ideals}\label{sec:transition-ideals}
The main results of this paper are concerned with \textit{transition ideals}.  A \textbf{transition ideal} is an ideal in the ring $\mathbb{Q}[X,X']$ for some set of variables $X$.  Transition ideals are not tied to the theory $\LIRR$, but can be seen as the vanishing ideals of transition formulas, and their operations can be understood in terms of corresponding operations on transition formulas.

For transition ideals $T_1$ and $T_2$, define \[T_1 \imul
T_2 \defeq \left(T_1[X' \mapsto X''] + T_2[X \mapsto X'']\right) \cap
\mathbb{Q}[X,X']\ .\]  
Equivalently, $T_1 \imul T_2$ is equal to $\I{\LIRR}{\iformula{T_1} \circ \iformula{T_2}}$.
For any transition ideal $T$ and natural number $n$,
define 
\[ T^0 = \langle \{x' - x : x\in X\}\rangle \qquad \iexp{T}{n} \defeq \underbrace{T \imul \dots \imul T}_{n \text{ times}} \qquad T^* \defeq \bigcap_{i=0}^\infty \iexp{T}{i}.\]
Note that, since the polynomials in a transition ideal are interpreted as \textit{constraints}, $T^*$ can be interpreted as the set of constraints that are common to all $T^i$.  In particular, if $F$ is a transition formula, then for any $n \in \mathbb{N}$, we have $F^n \models_{\LIRR} \iformula{(\I{\LIRR}{F})^*}$.

\begin{example}\label{ex:transitionIdeal}
    $T = \gideal{w' - wy, x' - 2x - y^2, y' - y - z, z' - 3z, y - z - 1}$ is an example of a transition ideal. $T^0 = \gideal{w' - w, x' - x, y' - y, z' - z}$, $T^1 = T$, and
    \begin{align*}
        T^2 = 
        \langle &w' - wy(y+z), x' - 2(2x+y^2) - (y+z)^2, y' - (y+z) - 3z, z' - 3(3z),\\ 
        &(y+z) - 3z - 1, y - z - 1 \rangle\\
        = \langle &w' - wy^2 - wz, x' - 4x - 3y^2 - 2yz - z^2, y' - y - 4z, z' - 9z, y - 2z - 1, y - z - 1\rangle\\
        = \langle &w' - w, x' - 4x - 3, y' - 1, z', y - 1, z\rangle\\
    \end{align*}
\end{example}

Define the \textbf{domain} of $T$ to be $\dom(T) \defeq T \cap \mathbb{Q}[X]$, and the
\textbf{invariant domain} of $T$ to be $\dom^*(T) = \sum_{n=0}^\infty
\dom(T^n)$. Informally, if we think of a transition ideal as a set of polynomial equations constraining the transition between a pre-state $X$ and post-state $X'$, then the domain of $T$ is the set of constraints that must be satisfied by a pre-state in order to have a successor. Pre-states that satisfy the invariant domain of $T$ are those states which have arbitrarily long computations described by $T$.
Given a transition ideal $T$, $\dom^*(T)$ can be calculated as follows. By inspection of the definition of $T^n\cdot T$, we see $\dom(T^n) \subseteq \dom(T^{n+1})$ for any $n\ge 1$. Therefore, we have the ascending chain of ideals:
\[
\dom(T) \subseteq \dom(T^2) \subseteq \dom(T^3) \subseteq \dots.
\]
By Hilbert's basis theorem this chain must stabilize at some $N$. That is, there exists $N\ge 1$ such that $\dom(T^N) = \dom(T^{N+1}) = \dom^*(T)$. 

We say that $T$ is \textbf{solvable} if there is a solvable
 $p : \mathbb{Q}[X] \rightarrow \mathbb{Q}[X]$ such that $x' - p(x) \in T$ for all $x \in X$; call $p$ a \textbf{solvability witness} for $T$.
We say that $T$ is \textbf{ultimately solvable} if $T + \dom^*(T)$ is solvable.

\begin{example}\label{ex:solvabelTransitionIdeal}
    Recall \cref{ex:transitionIdeal}, with $T = \gideal{w' - wy, x' - 2x - y^2, y' - y - z, z' - 3z, y - z - 1}$. $\dom(T) = \gideal{y - z - 1}$. For this example, the invariant domain stabilizes at $N = 2$, and so $\dom(T^2) = \dom^*(T) = \gideal{y - 1, z}$. $T$ is \emph{not} a solvable transition ideal. This is because the dynamics of $w$, $w\mapsto wy$ cannot be captured by a solvable polynomial map.

    $T' = \gideal{x' - 2x - y^2, y' - y - z, z' - 3z, y - z - 1}$ is a solvable transition ideal recognized by the solvability witness $p:\mathbb{Q}[x, y, z] \rightarrow\mathbb{Q}[x, y, z]$ defined by $p(x) = 2x + y^2$, $p(y) = y+z$, and $p(z) = 3z$. Note that even though there is a non-linear dependence on the variable $y$ for the assignment $p(x)$, $p$ is still solvable using the variable partition $\set{y, z} \cup \set{x}$.

    $T$ is an example of an ultimately solvable transition ideal. $T + \dom^*(T) =\langle w' - w, x' - 2x - 1, y' - 1, z', y - 1, z \rangle$ is solvable, and is recognized by the solvability witness $q:\mathbb{Q}[w, x, y, z]\rightarrow\mathbb{Q}[w, x, y, z]$ defined by $q(w) = w$, $q(x) = 2x + 1$, $q(y) = 1$, and $q(z) = 0$.
\end{example}

Let $T \subseteq \mathbb{Q}[X,X']$ and $U \subseteq \mathbb{Q}[Y,Y']$ be
transition ideals (possibly over different variables). 
A polynomial homomorphism $s : \mathbb{Q}[Y] \rightarrow \mathbb{Q}[X]$ is a \textbf{simulation} from $T$ to $U$ (notice reversal of direction) if for every polynomial $p \in U$, we have $\dvext{s}(p) \in T$,
where $\dvext{s} : \mathbb{Q}[Y,Y'] \rightarrow \mathbb{Q}[X,X']$ denotes the extension of $s$ to the ``doubled'' vocabulary, which maps each $x \in X$ to $s(x)$ and each $x' \in X'$ to $s(x)'$. We say that $s$ is a \textbf{linear simulation} if it is linear and a simulation.
\newcommand{\simcompose}[2]{#1 ; #2}
If $T \subseteq \mathbb{Q}[X,X']$, $U \subseteq \mathbb{Q}[Y,Y']$,
  and $V \subseteq \mathbb{Q}[Z,Z']$ are transition ideals and $s : T \rightarrow U$ and $t : U \rightarrow V$ are simulations, then their composition $\simcompose{t}{s} \defeq s \circ t$ is a simulation from $T$ to $V$ (again noting that simulations go in the opposite direction of polynomial homomorphisms).


\section{Solvable Reflections} \label{sec:reflection}

In this section, we show that every transition ideal $T \subseteq \mathbb{Q}[X,X']$ admits a \textit{solvable reflection}: there is a solvable transition ideal $U$ and a linear simulation $s : T \rightarrow U$ that is a closer approximation of $T$ than any other simulating solvable transition ideal.

\begin{example}
\Cref{fig:reflection} illustrates a transition ideal $T$ that will be used as a running example throughout this section.  One may think of $T$ as the polynomial map 
\[ f(w,x,y,z) = (z^2, 4x(1-x), -4x^2 + 3x + y - 1, z + x^2 - 2xy + y^2) \]
(corresponding to the first four generators of $T$) restricted to the domain $w - z^2 - 1$ (corresponding to the fifth).  The map $f$ is not solvable, since $x$ exhibits a non-linear self-dependence
(in fact, $x \mapsto 4x(1-x)$ is a logistic map, a famous example in chaos theory that \citet{logistic} suggested as the basis of a pseudo-random number generator). If we restrict the transition ideal to the variable $w$, the resulting transition ideal $T_w = \gideal{w' - w + 1}$ \textit{is} solvable, and we can compute a closed form for its $i$th iterate: $T_w^i = \gideal{w' - w + i}$.  This is an instance of a \textit{solvable abstraction} of $T$: $T_w$ is a solvable transition ideal that approximates the dynamics of $T$, and the nature of the approximation is given by the inclusion homomorphism $\mathbb{Q}[w] \rightarrow \mathbb{Q}[w,x,y,z]$ (mapping $w \mapsto w$).

There are other solvable abstractions of $T$ which capture different aspects of $T$'s dynamics. Observe that while it's challenging to understand the dynamics of $x$ or $y$, their \textit{difference} behaves predictably: $T$ contains the polynomial
$(x' - y') - (x - y) + 1$, indicating that the value of $(x-y)$ decreases by 1 at each step. Coincidentally, the dynamics of $(x-y)$ are identical to that of $w$ 
(both decrease by 1), so this information can be represented as the solvable abstraction $\tuple{\set{w \mapsto x-y},T_w}$.

Yet another solvable abstraction $\tuple{s,U}$ is pictured in \cref{fig:reflection}.  This abstraction is more desirable than either $\tuple{\set{w \mapsto w}, T_w}$ or $\tuple{\set{w \mapsto x-y}, T_w}$, in the sense that $\tuple{s,U}$ captures the dynamics of not only $(x-y)$ and $w$, but also $z$.  In fact, $\tuple{s,U}$ is a \textit{solvable reflection} of $T$, in the sense that any other solvable abstraction of $T$ similarly factors through $\tuple{s,U}$.
\end{example}

Formally, a \textbf{solvable reflection} of $T \subseteq \mathbb{Q}[X,X']$ (with respect to linear simulations) is a pair $\tuple{s,U}$ consisting of a transition ideal $U \subseteq \mathbb{Q}[Y,Y']$ (for some set of variables $Y$) and a polynomial homomorphism
$s : \mathbb{Q}[Y] \rightarrow \mathbb{Q}[X]$ such that:
\begin{enumerate}
\item $s$ is a linear simulation from $T$ to $U$
\item $U$ is solvable
\item For any other pair $\tuple{v,V}$ satisfying 1 and 2, there exists a unique linear simulation $\hat{v} : U \rightarrow V$ such that $v = \simcompose{\hat{v}}{s}$.
\end{enumerate}
\Cref{sec:reflection-computation} describes a procedure for computing solvable reflections.  We then extend this result in two ways: (1) \cref{sec:ultimately-solvable} generalizes from solvable to \textit{ultimately} solvable reflections,
and (2) \cref{sec:polynomial-simulations} generalizes from linear simulations to polynomial simulations.

\subsection{Computing Solvable Reflections} \label{sec:reflection-computation}


\begin{figure}
    \begin{subfigure}{0.25\textwidth}
        \[
        \left(\begin{array}{l}
        w' - z^2\\
        x' - 4x(1-x)\\
        y' + 4x^2 - 3x - y - 1\\
        z' - z - x^2 + 2xy - y^2\\
        w - z^2 - 1
            \end{array}\right)
        \]
        \caption{A transition ideal, $T$ \label{fig:reflection-T}}
    \end{subfigure}
    \begin{subfigure}{0.45\textwidth}
        \[
        \tuple{
        \left\{\begin{array}{rl}
        a &\mapsto (x-y)\\
        b &\mapsto w\\
        c &\mapsto z
        \end{array}\right\},
        \left(\begin{array}{l}
        a' - a - 1\\
        b' - b + 1\\
        c' - c - a^2\\
        b - c^2 - 1
            \end{array}\right)}
        \]
        \caption{A solvable reflection $\tuple{s,U}$ of $T$ \label{fig:reflection-U}}
    \end{subfigure}
        \begin{subfigure}{0.25\textwidth}
        \[
                \left(\begin{array}{l}
        (x' - y') - (x - y) + 1\\
        w' - w + 1\\
        z' - z - (x-y)^2\\
        w - z^2 - 1\\
            \end{array}\right)
        \]
        \caption{Image of $U$ under $\dvext{s}$  \label{fig:reflection-image}}
    \end{subfigure}
    \caption{A solvable reflection of a transition ideal \label{fig:reflection}}
\end{figure}

\newcommand{\Stratum}[2]{\textsf{Stratum}\left(#1,#2\right)}

To begin, we give an alternate characterization of solvable transition ideals, which will serve as the basis of our algorithm for computing solvable reflections.
For a transition ideal $T \subseteq \mathbb{Q}[X,X']$ and a set of polynomials $P \subseteq \mathbb{Q}[X]$, define 
\[
\Det{T}{P} \defeq \set{ p \in \gspace{X} : \exists q \in P.  p' - q \in T  }
\]
Intuitively, $\Det{T}{P}$ represents the set of linear functionals that are ``determined up to $P$''; i.e., $T$ constrains the post-state value of $p \in \Det{T}{P}$ to be equal to some $q \in P$.
Observe that $T$ is solvable exactly when there is an ascending chain $X_1 \subset X_2 \dots \subset X_n = X$ of sets of variables (called a \textit{stratification} of $X$) such that for each $i \in \set{1, \dots, n-1}$, we have
$\gspace{X_{i+1}} \subseteq \Det{T}{\gspace{X_{i+1}} + \galg{X_i}}$.
(i.e., for each $x \in X_{i+1}$, we have $x' - p - q \in T$ for some linear polynomial $p \in \gspace{X_{i+1}}$
over the variables $X_{i+1}$ and some polynomial $q \in \galg{X_i}$ over the variables $X_i = X_1 \cup \dots \cup X_i$).

\begin{example}
Consider the solvable transition ideal $U$ in \cref{fig:reflection-U}.  A solvability witness for $U$
is the polynomial homomorphism $f_U$ that sends $a \mapsto a-1$, $b \mapsto b+1$, and $c \mapsto c+a^2$, and the corresponding partition of the variables is $\set{a,b} \cup \set{c}$.  The fact that $f_U$ is solvable corresponds precisely to the facts that
both $f(a)$ and $f(b)$ belong to $\gspace{\set{a,b}} + \galg{\emptyset}$,
and $f(c)$ belongs to $\gspace{\set{a,b,c}} + \galg{\set{a,b}}$.

We may derive from this solvability witness a stratification
$\set{a,b} \subset \set{a,b,c}$ (i.e., the $i$th stratum is the union of the first $i$ cells in the ordered partition).  Observe that
$\Det{U}{\gspace{\set{a,b}} + \galg{\emptyset}} = \gspace{\set{a,b}}$,
and $\Det{U}{\gspace{\set{a,b,c}} + \galg{a,b}} = \gspace{\set{a,b,c}}$.

Furthermore, observe that the simulation $u : T \rightarrow U$ induces a corresponding structure in $T$, namely
$\Det{T}{\gspace{\set{(x-y),w}} + \galg{\emptyset}} = \gspace{\set{(x-y),w}}$ (where $\set{(x-y),w)} = \set{u(a),u(b)}$)
and $\Det{U}{\gspace{\set{(x-y),w,z}} + \galg{(x-y),w}} = \gspace{\set{(x-y),w,z}}$ (where $\set{(x-y),w,z)} = \set{u(a),u(b),u(c)}$.  In fact this stratification structure determines the solvable reflection of $T$ uniquely (up to isomorphism) in the sense that if $v : \mathbb{Q}[d_1,d_2,d_3] \rightarrow \mathbb{Q}[w,x,y,z]$ is any polynomial homomorphism such that
$\set{v(d_1),v(d_2)}$ is a basis for $\gspace{(x-y),w}$ and
$\set{v(d_1),v(d_2),v(d_3)}$ is a basis for $\gspace{(x-y),w,z}$, then
$\tuple{v,\invimage{v}{T}}$ is a solvable reflection of $T$.
The essential idea behind the algorithm presented in this section is to calculate this structure by discovering each stratum in turn.
\end{example}

\newcommand{\stratum}[2]{\mathcal{S}_{#1}(#2)}
\newcommand{\strata}[1]{\mathcal{S}^*(#1)}

\begin{theorem} \label{thm:solvable-reflection-exist}
 Every transition ideal has a solvable reflection.
\end{theorem}
\begin{proof}
Let $T \subseteq \mathbb{Q}[X,X']$ be a transition ideal.  We may calculate the solvable reflection of $T$ as follows.
For each natural number $i$, we define a linear subspace of polynomials $\stratum{i}{T} \subseteq \gspace{X}$ as follows:
\begin{itemize}
\item $\stratum{0}{T} \defeq \set{0}$.
\item For each $i \geq 0$, $\stratum{i+1}{T}$ is the greatest fixed point of $S \mapsto \Det{T}{S + \galg{\stratum{i}{T}}}$.  Such a fixed point always exists by the Knaster-Tarski fixed point theorem, noting that $S \mapsto \Det{T}{S + \galg{\stratum{i}{T}}}$ is a monotone operator on the complete lattice of subspaces of $\gspace{X}$.
\end{itemize}
Since $\stratum{0}{T} \subseteq \stratum{1}{T} \subseteq \dots$ is an ascending chain of subspaces of a finite-dimensional space $\gspace{X}$, it must stabilize (i.e., there is some $n$ such that $\stratum{n}{T} = \stratum{n+1}{T}$).  Call the resulting space $\strata{T}$.

For any $i \in \mathbb{N}$, let $d_i$ be the dimension of $\stratum{i}{T}$.  Choose an ordered basis $p_1,\dots,p_n$ for 
$\strata{T}$ such that for each $i$, $p_1,\dots,p_{d_i}$ spans
$\stratum{i}{T}$ (such a basis may be obtained by choosing an arbitrary basis for $\stratum{0}{T}$ and then extending it to $\stratum{1}{T}$, then extending that basis to $\stratum{2}{T}$, and so on).  
Let $Y = \set{y_{1},\dots,y_{n}}$ be a set of variables disjoint from $X$, and let $u : \mathbb{Q}[Y] \rightarrow \mathbb{Q}[X]$ be the homomorphism that maps $y_{i} \mapsto p_{i}$.  Finally, define $R(T) \defeq \tuple{u, \invimage{u}{T}}$.  We will show that $R(T)$ is a solvable reflection of $T$.

First, we show that $\invimage{u}{T}$ is solvable.  We construct a solvability witness as follows.
For any $i \leq n$, let $s(i)$ be the least number such that $i \leq d_{s(i)}$.
For each $1 \leq k \leq n$, let $Y_k = \set{ y_i : s(i) = k }$.
For any $i$ we have
$p_{i} \in \stratum{s(i)}{T}$, and so
\[ p_{i} \in  \Det{T}{\stratum{s(i)}{T} + \galg{\stratum{s(i)-1}{T}}}
 \]
 It follows that there is some $q_i \in \stratum{s(i)}{T}$ and $\hat{q}_i \in \galg{\stratum{s(i)-1}{T}}$ such that $p_{i}' - q_i - \hat{q}_i \in T$.
 Since $u[\textit{span}(Y_i)] = \stratum{s(i)}{T}$
 and
 $u[\galg{Y_1 \cup \dots \cup Y_{i-1}}] = \galg{\stratum{s(i)-1}{T}}$,
 there is some $t_i \in \textit{span}(Y_i)$
 and $\hat{t}_i \in \galg{Y_1 \cup \dots \cup Y_{i-1}}$ such that
 $u(t_i) = q_i$ and $u(\hat{t}_i) = \hat{q}_i$.
 Define $r : \mathbb{Q}[Y] \rightarrow \mathbb{Q}[Y]$ to be the polynomial homomorphism that sends
 $y_i \mapsto t_i + \hat{t}_i$.
 Observe that $r$ is a solvability witness for $\invimage{u}{T}$: $r$ is solvable by construction,
 and for each $y_i$ we have 
 \[ \dvext{u}(y_i' - r(y_i)) = \dvext{u}(y_i' - t_i - \hat{t}_i) = \dvext{u}(y_i') - \dvext{u}(t_i) - \dvext{u}(\hat{t}_i) = p_i' - q_i - \hat{q}_i \in T
 \]
 and so $y_i' - r(y_i) \in \dvext{u}^{-1}[T]$.
 
 Next, we must show that $R(T)$ is universal.  Suppose that
 $V \subseteq \mathbb{Q}[Z,Z']$ is a solvable transition ideal, and that $v : \mathbb{Q}[Z] \rightarrow \mathbb{Q}[X]$ is a linear simulation from $T$ to $V$.  We must show that there is a homomorphism $\hat{v} : \mathbb{Q}[Z] \rightarrow \mathbb{Q}[Y]$ such that $u \circ \hat{v} = v$ and $\hat{v}$ is a simulation from $\invimage{u}{T}$ to $V$.

 It is sufficient to show that for each $z \in Z$, we have
 $v(z) \in \strata{T}$: under this assumption, for each $z$, $v(z)$ can be written (uniquely) as a linear combination
$v(z) = a_1p_1 + \dotsi + a_np_n = u(a_1y_1 + \dots + a_ny_n)$, and so we can define $\hat{v}$ by $\hat{v}(z) \defeq a_1y_1 + \dots + a_ny_n$.
It follows that
\[\dvext{\hat{v}}^{-1}[\dvext{u}^{-1}[T]] = \dvext{(u \circ \hat{v})}^{-1}[T] = \dvext{v}^{-1}[T] \supseteq V\ ,\] and so $\hat{v}$ is a simulation from $\dvext{u}^{-1}[T]$ to $V$.

 Since $U$ is solvable, there is a partition $Z_1,\dots,Z_k$ of $Z$ and a polynomial homomorphism $g : \mathbb{Q}[Z] \rightarrow \mathbb{Q}[Z]$ such that for each $i$ and each
 $z \in Z_i$, $g(z)$ can be written as the sum of a linear term in $Z_i$ and a polynomial in $Z_1,...,Z_{i-1}$.   For each $i \in \set{0,\dots,n}$, let $Z_{\leq i}$ denote $\bigcup_{j = 1}^i Z_j$.  We show by induction on $k$ that for all $z \in Z_{\leq k}$, $v(z) \in \stratum{k}{T}$.  The base case $k=0$ is trivial: $Z_{\leq 0}$ is empty.  For the inductive step, suppose 
 $s(Z_{\leq k}) \subseteq \stratum{k}{T}$, and prove
 $s(Z_{\leq {k+1}}) \subseteq \stratum{k+1}{T}$.
  Let $z \in Z_{\leq_{k+1}}$.  Since $g$ is a witness to solvability of $V$, we have $z' - g(z) \in V$, and since
  $v$ is a simulation from $T$ to $V$, we have
  $\dvext{v}(z' - g(z)) \in T$.  Since $g$ is solvable, we have
 $g(z) \in \gspace{Z_{k+1}} + \galg{Z_{\leq k}}$, and so
 $v(g(z)) \in \gspace{v(Z_{k+1})} + \galg{v(Z_{\leq k})}$.
 Since $\dvext{v}(z' - g(z)) = v(z)' - v(g(z)) \in T$ and $v(g(z)) \in \gspace{v(Z_{k+1})} + \galg{v(Z_{\leq k})}$, we have
 $v(z) \in \Det{T}{\gspace{v(Z_{k+1})} + \galg{v(Z_{\leq k})}}$, and so by the induction hypothesis and monotonicity of $\textsf{Det}$ we have $v(z) \in \Det{T}{\gspace{v(Z_{k+1})} + \stratum{k}{T}}$.  Since this holds for all $z \in Z_{k+1}$, we have
 $\gspace{v(Z_{k+1}} \subseteq \Det{T}{\gspace{v(Z_{k+1})} + \stratum{k}{T}}$, and since $\stratum{k+1}{T}$ is defined to be greatest fixedpoint of
 $S \mapsto \Det{T}{S + \stratum{k}{T}}$, we have 
 $\gspace{v(Z_{k+1})} \subseteq \stratum{k+1}{T}$, and so
 for all $z \in Z_{k+1}$, $v(z) \in \stratum{k+1}{T}$.
\end{proof}

The proof of \cref{thm:solvable-reflection-exist} is constructive and gives rise to a procedure for computing solvable reflections of transition ideals, depicted in \cref{alg:solvable-reflection}.   The procedure relies on a subroutine for computing $\textsf{Det}$, which is given in
\cref{alg:det} (the correctness of which is \cref{lem:det}).   Otherwise, the procedure follows the steps of the proof directly. 

\begin{algorithm}[t]
\KwIn{Finite sets of polynomials $T \subseteq \mathbb{Q}[X,X']$,
$V \subseteq \mathbb{Q}[X]$, and $Q \subseteq \mathbb{Q}[X]$}
\KwOut{Set of polynomials $D$ such that $\gspace{D} = \Det{\gideal{T}}{\gspace{V}+\galg{Q}}$}
$Y \gets \set{ y_v : v \in V }$ \tcc*{Introduce a variable $y_v$ for each generator $v \in V$}
$Z \gets \set{ z_q : q \in Q }$ \tcc*{Introduce a variable $z_q$ for each generator $q \in Q$}
Let $f : \mathbb{Q}[X,Y,Z] \rightarrow \mathbb{Q}[X,X']$ be the map that sends $x \mapsto x'$, $y_v \mapsto v$, $z_q \mapsto q$\;
$F \gets \textit{inv.image}(f,T)$\;
$G \gets $  Gr\"{o}bner basis for $F$ w.r.t.  $\elimorder{X \cup Y}$ for some graded order $\ll$\;
$D \gets \emptyset$\;
\ForEach{$g \in G$}{
    \If{$g = p + q + r$ for some non-zero $p \in \gspace{X}$, $q \in \gspace{Y}$, and $r \in \mathbb{Q}[Z]$}{
        $D \gets D \cup \set{ p }$\;
    }
}
\Return{$D$}
\caption{Computation of determined functionals \label{alg:det}}
\end{algorithm}

\begin{lemmarep} \label{lem:det}
Let $T \subseteq \mathbb{Q}[X,X']$,
$V \subseteq \mathbb{Q}[X]$, and $Q \subseteq \mathbb{Q}[X]$ (for some set of variables $X$).  Then \cref{alg:det} computes a set of polynomials $D$ such that $\gspace{D} = \Det{\gideal{T}}{\gspace{V}+\galg{Q}}$.
\end{lemmarep}
\begin{appendixproof}
Let $Y, Z, f, F, G$, and $D$ be as in \cref{alg:det}.

First we prove $D \subseteq \Det{\gideal{T}}{\gspace{V}+\galg{Q}}$. 
Let $p \in D$. 
Then there is some $q \in \gspace{Y}$ and $r \in \mathbb{Q}[Z]$ such that such that $p - q - r \in G$.  It follows that $f(p - q - r) = p' - f(q) - f(r) \in \gideal{T}$.
Finally, observe that $f(q) \in f(\gspace{Y}) = \gspace{V}$, and that  $f(r) \in f(\mathbb{Q}[Z]) = \galg{Q}$, and thus $p \in \Det{\gideal{T}}{\gspace{V}+\galg{Q}}$.

Next we prove $\Det{\gideal{T}}{\gspace{V}+\galg{Q}} \subseteq \gspace{D}$.  Let $p \in \Det{\gideal{T}}{\gspace{V}+\galg{Q}}$, and suppose that $p$ is non-zero.  Then there is some $v \in \gspace{V}$ and some $q \in \galg{Q}$ such that $p' - v - q \in \gideal{T}$.
Since $v \in \gspace{V} = f(\gspace{Y})$, there is some $\hat{v} \in \gspace{Y}$ such that $f(\hat{v}) = v$.  Since $q \in \galg{Q} = f(\mathbb{Q}[Z])$, there is some polynomial $\hat{q} \in \mathbb{Q}[Z]$ such that $f(\hat{q}) = q$. 
Since $\gideal{T}$ contains $p' - v - q$, $\gideal{F} = \gideal{\textit{inv.image}(f,T)} = f^{-1}[T]$ must contain $p - \hat{v} - \hat{q}$.

Let $G = \set{ g_1, \dots, g_n}$ be a Gr\"{o}bner basis for $\gideal{F}$ with respect to $\elimorder{X \cup Y}$, and let
$I$ be the subset of $\set{1, \dots, n}$ such that $g_i = p_i + q_i + r_i$ for some non-zero linear $p_i \in \mathbb{Q}[X]$,
linear $q_i \in \mathbb{Q}[Y]$, and $r_i \in \mathbb{Q}[Z]$,
so that $D = \set{ p_i : i \in I }$.
Since $G$ is a  Gr\"{o}bner basis for $\gideal{F}$ with respect to $\elimorder{X \cup Y}$ and $p - \hat{v} - \hat{q} \in \gideal{F}$,
we have 
$p - \hat{v} - \hat{q} = \sum_{i=0}^n c_ig_i$, for some $c_1,\dots,c_n \in \mathbb{Q}[X,Y,Z]$ with
$\LM(c_ig_i) \elimorder{X\cup Y} \LM(p - \hat{v} - \hat{q})$ for each $i$.   Since $\elimorder{X \cup Y}$ is using a graded order and $p - \hat{v} - \hat{q}$ is linear in $X$ and $Y$, so must be each $c_ig_i$, and so $c_i \in \mathbb{Q}$ for each $i \in I$.
It follows that $p \in \gspace{ \set{ p_i : i \in I } } = \gspace{D}$.
\end{appendixproof}

\begin{algorithm}[t]
\KwIn{Finite set of polynomials $T \subseteq \mathbb{Q}[X,X']$}
\KwOut{Solvable reflection of $T$}
$i \gets 0$\;
$n \gets 0$\;
$S_0 \gets \emptyset$\;
\Repeat{$n = n'$}{
  $n' = n$\;
  $S_{i+1} \gets X$\;
  \tcc{$\gspace{S_{i+1}}$ is the greatest fixpoint of $X \mapsto \Det{T}{\gspace{X} + \galg{S_i}}$}
  \Repeat{$\textit{dim} = |S_{i+1}|$}{
    $\textit{dim} \gets |S_{i+1}|$\;
    $S_{i+1} \gets \Det{T}{\gspace{S_{i+1}}+\galg{S_i}}$ \tcc*{\cref{alg:det}}
  }
  \ForEach{$p \in S_{i+1}$}{
    \If{$p \notin \gspace{q_0, \dots, q_{n-1}}$}{
        $q_n \gets p$\;
        $n \gets n + 1$\;
    }
  }
}
Let $u$ be the map $\mathbb{Q}[y_0,\dots,y_{n-1}] \rightarrow \mathbb{Q}[X]$ mapping $y_i \mapsto q_i$\;
\Return{$\tuple{u, \textit{inv.image}(\dvext{u},T)}$}
\caption{Solvable reflection \label{alg:solvable-reflection}}
\end{algorithm}

\subsection{Ultimately Solvable Reflections} \label{sec:ultimately-solvable}

The algorithm presented in \cref{sec:closure} reveals that a weaker condition than solvability is sufficient in order to compute the Kleene closure of a transition ideal, namely the transition needs to be \textit{ultimately solvable}.  This raises the question of whether it's possible to compute ultimately solvable reflections of arbitrary transition ideals, and thereby obtain a more powerful algorithm for generating polynomial invariants for loops.
In this section, we answer that question in the affirmative.

 Let $T$ be a transition ideal.  Define a sequence
 $\tuple{t_0,T_0},\tuple{t_1,T_1},\dots$ where $t_0$ is the identity function, $T_0 = T$, 
 and for each $i \geq 0$,
 $t_{i+1}$ is the simulation component of the solvable reflection of $T_i+\dom^*(T_i)$, and $T_{i+1}$ is $\invimage{t_{i+1}}{T_i}$.
 Since for all $i$, if $t_i$ is \text{not} invertible then the dimension of $T_i$ is strictly smaller than $T_{i-1}$, there must be some first index $n$ such that
 $t_n$ is invertible.  Define
 $R^*(T) \defeq \tuple{u^*,\dvext{u^*}^{-1}[T]}$, where
 $u^* = t_{1} \circ \dots \circ t_{n-1}$.
\begin{lemmarep}
Let $T$ be a transition ideal.  Then $R^*(T)$ is an ultimately solvable reflection of $T$.
\end{lemmarep}
\begin{appendixproof}
    Observe that $\dvext{u^*}^{-1}[T]$ is ultimately solvable, since 
    (by the definition of $u^*$), the solvable reflection of
    $\dvext{u^*}^{-1}[T] + \dom^*(\dvext{u^*}^{-1}[T])$ is isomorphic
    to $\dvext{u^*}^{-1}[T] + \dom^*(\dvext{u^*}^{-1}[T])$.
    
    Towards universality, we show that for all $i$, we have
    \begin{enumerate}
        \item  $t_0 \circ \dots \circ t_i$ is a simulation $T \rightarrow T_i$
        \item For any ultimately solvable $V$ and simulation $v : T \rightarrow V$, there is a unique simulation $v_i : T_i \rightarrow V$ such that $v = t_0 \circ \dots \circ t_i \circ v_i$.
    \end{enumerate}
    by induction on $i$.  The base case $i=0$ is trivial.
    For the induction step, suppose that (1) and (2) hold for $i$.
    By definition, $t_{i+1}$ is the simulation component of a solvable reflection of $R(T_i + \dom^*(T_i))$, and $T_{i+1}$ is $\dvext{t_{i+1}}^{-1}[T_i]$.  Thus,
    $t_{i+1}$ is a simulation from $T_i$ to $T_{i+1}$.  Since
    $t_0 \circ \dots \circ t_i$ is a simulation $T \rightarrow T_{i}$, it follows that (1) the composition $t_0 \circ \dots \circ t_i \circ t_{i+1}$ is a simulation from $T$ to $T_{i+1}$.
    Next, suppose that $V$ is ultimately solvable and $v : T \rightarrow V$ is a simulation.
    By the induction hypothesis,
    there is a unique simulation $v_i : T_i \rightarrow V$ such that $v = t_0 \circ \dots \circ  t_i \circ v_i$.  
    It follows that $v_i$ is also a simulation
    from $T_i + \dom^*(T_i)$ to $V + \dom^*(V)$.
    Since $V$ is ultimately solvable,
    $V + \dom^*(V)$ is solvable.  
    Since there is some $W_{i+1}$ such that
    $\tuple{t_{i+1},W_{i+1}}$ is a solvable reflection of $T_i+\dom^*(T_i)$, 
    $V + \dom^*(V)$ is solvable, and $v_i : T_i + \dom^*(T_i) \rightarrow V + \dom^*(V)$ is a simulation,
    there is a unique simulation $v_{i+1}$ from $W_{i+1}$ to $V + \dom^*(V)$ such that $t_{i+1} \circ v_{i+1} = v_i$. We have
    \begin{align*}
        V &\subseteq \invimage{v_i}{T_i} & \text{$v_i$ a simulation $T_i \rightarrow V$}\\
        &= \invimage{t_{i+1} \circ v_{i+1}}{T_i} & t_{i+1} \circ v_{i+1} = v_i\\
        &= \invimage{v_{i+1}}{\invimage{t_{i+1}}{T_i}} \\
        &= \invimage{v_{i+1}}{T_{i+1}} & \text{Definition of } T_{i+1}
    \end{align*}
    and thus $v_{i+1}$ is a simulation from $T_{i+1}$ to $V$.
    For uniqueness, suppose $v_{i+1}'$ is a simulation from $T_{i+1}$ to $V$ with $v = t_0 \circ \dots \circ t_{i+1} \circ v_{i+1}'$.  Since
    $t_{i+1} \circ v_{i+1}'$ is a simulation $T_i \rightarrow V$ with
    $t_0 \circ \dots \circ t_i \circ (t_{i+1} \circ v_{i+1}') = v$, and
    $v_i$ is the unique such simulation, we have
    $v_i = t_{i+1} \circ v_{i+1}'$.  Since $v_{i+1}$ is unique such that
    $v_i = t_{i+1} \circ v_{i+1}$, we have $v_{i+1} = v_{i+1}'$.

    Finally we show that $\tuple{u^*,\dvext{u^*}^{-1}[T]}$ is universal.
    Let $V$ be ultimately solvable, and let $v : T \rightarrow V$ be a simulation.  By (2), there is a unique simulation $v_{n-1} : T_{n-1} \rightarrow V$ such that $v = t_0 \circ \dots \circ t_{n-1} \circ v_i = u^* \circ v_i$.  Since $T_{n-1} = \dvext{u^*}^{-1}[T]$, we have the result.
\end{appendixproof}

\subsection{Polynomial Simulations}
\label{sec:polynomial-simulations}

Here, we consider a generalization of our definition of (ultimately) solvable reflections, in which the simulation from a transition ideal to its reflection is a polynomial map rather than a linear map.

Let $X$ be a set of variables and let $d \in \mathbb{N}$ be a fixed degree bound.  Let $X^{\leq d}$ be the set of monomials of degree at most $d$ (excluding 1), let $Y$ be a set of variables of cardinality equal to that of $X^{\leq d}$, and let $f_{X,d} : Y \rightarrow X^{\leq d}$ be a bijection.  Observe that if $Z$ is a set of variables and $g : \mathbb{Q}[Z] \rightarrow \mathbb{Q}[X]$ is a polynomial homomorphism of degree at most $d$, then there is unique \textit{linear} polynomial homomorphism $\hat{g}$ such that $g = \hat{g} \circ f_{X,d}$.  As a result, we can reduce the problem of computing reflections with respect to bounded-degree polynomial simulations to the problem of computing reflections with respect to linear simulations:

\begin{lemmarep}
    Let $T \subseteq \mathbb{Q}[X,X']$ be a transition ideal, and
     let $d \in \mathbb{N}$ be a fixed degree bound.
    Suppose that $\tuple{u,U}$ is an (ultimately) solvable reflection of $\invimage{f_{X,d}}{T}$.  Then
    $\tuple{u \circ f_{X,d},U}$ is an (ultimately) solvable reflection of $T$ with respect to degree-$d$ simulations, in the sense that (1) $u \circ f_{X,d}$ has degree at most $d$, (2) $U$ is solvable,
    and (3) for solvable transition ideal $V$ and simulation $v$ from $T$ to $V$ of degree at most $d$, there is a unique linear simulation $\hat{v}$ such that $v = \hat{v} \circ u \circ f_{X,d}$.
\end{lemmarep}
\begin{appendixproof}
    Since $u$ is linear, the polynomial homomorphism $u \circ f_{X,d}$ is degree-$d$, and so $\tuple{u \circ f_{X,d},U}$ is a degree-$d$ solvable abstraction of $T$.  It remains only to show that is satisfies the desired universal property.

    Suppose $V \subseteq \mathbb{Q}[Z,Z']$ is an (ultimately) solvable transition ideal, and that $v : T \rightarrow V$ is a degree-$d$ simulation.  Then there  exists a unique linear simulation $\hat{v} : \invimage{f_{X,d}}{T} \rightarrow V$ such that $v = \hat{v} \circ f_{X,d}$.  Since  $U$ is an (ultimately) solvable reflection of $\invimage{f_{X,d}}{T}$,
    $\hat{v}$ is a linear simulation, and $V$ is (ultimately) solvable, there is a unique linear simulation $w : U \rightarrow V$ such that $w \circ u = \hat{v}$.  Finally, observe that $w \circ (u \circ f_{X,d}) = (w \circ u) \circ f_{X,d} = \hat{v} \circ f_{X,d} = v$.
\end{appendixproof}

\section{Kleene closure of solvable transition ideals} \label{sec:closure}
In this section we describe how to compute $T^* = \bigcap_{i=0}^\infty T^i$ when $T$ is either a solvable and ultimately solvable transition ideal. In doing so we introduce a sub-problem of potential independent interest. The sub-problem asks how to find the set of rational polynomials that evaluate to 0 for every position in a $\mathbb{Q}$-algebra sequence defined by a solvable polynomial map. 

\subsection{Finding the Relations of a Solvable Map over a $\mathbb{Q}$-algebra}
\begin{problem}\label{prob:solvePolyMap}
    Let $A$ be a $\mathbb{Q}$-algebra, let $X$ be a finite set of variables, and let $v \in A^X$.  \textit{Given a solvable map $f : \mathbb{Q}[X] \rightarrow \mathbb{Q}[X]$ and basis $I$ such that  $\gideal{I} = \I{A}{\set{v}}$, find a basis for $\I{A}{\set{f^i_A(v):i \in \mathbb{N}}} \subseteq \mathbb{Q}[X]$.}
\end{problem}
Intuitively, the solvable map $f$ in \cref{prob:solvePolyMap} defines a  sequence, $(v, f^1_A(v), f^2_A(v), \dots)$. The goal of the problem is to find a basis for the set of polynomials $p\in \mathbb{Q}[X]$ such that \[(p^A(v), p^A(f^1_A(v)), p^A(f^2_A(v)), \dots) = (0, 0, 0, \dots).\]
The purpose of the ideal $\gideal{I} =\I{A}{\set{v}}$ is to give the set of polynomial relations of the first element of the sequence, $v$. In the case of \cref{prob:solvePolyMap} we take $\I{A}{\set{v}}$ as a given to encode the relevant information of $A$ and $v$.
\begin{example}\label{ex:solveSolvable}
    Let $X = \{x, y\}$, $A = \mathbb{Q}[w]/\langle w^2-3\rangle$, and $v = \{x\mapsto w + \gideal{w^2-3} , y \mapsto 2w + 3 + \gideal{w^2-3}\}$. For this case we have $\I{A}{\set{v}} = \langle 2x - y + 3, x^2-3\rangle$. Let $f : \mathbb{Q}[x, y] \rightarrow \mathbb{Q}[x, y]$ be the polynomial homomorphism defined by $f(x) = 2y$ and $f(y) = 2x$. Then $f$ is a solvable map that defines the following sequence over $A^X$:
    \[
    \left(\begin{Bmatrix}
        x\mapsto w + \gideal{w^2-3}\\
        y \mapsto 2w + 3 + \gideal{w^2-3}
    \end{Bmatrix}, 
    \begin{Bmatrix}
        x\mapsto 4w + 6+ \gideal{w^2-3}\\
        y \mapsto 2w+ \gideal{w^2-3}
    \end{Bmatrix}, 
    \begin{Bmatrix}
        x\mapsto 4w+ \gideal{w^2-3}\\
        y \mapsto 8w + 12+ \gideal{w^2-3}
    \end{Bmatrix},
    \dots
    \right)
    \]
    It can readily be verified that $p(x, y) = x^2 - 4xy + y^2\in\I{A}{\set{f^i_A(v):i \in \mathbb{N}}}$. For instance, let $v_2 =\set{x \mapsto 4w+ \gideal{w^2-3}, y\mapsto 8w+12 + \gideal{w^2-3}}$ $\in A^X$ denote the second (indexing from 0) valuation of the sequence. Then $p^A(v_2) = (4w+ \gideal{w^2-3})^2 - 4(4w + \gideal{w^2-3})(8w+12+ \gideal{w^2-3}) + (8w+12+ \gideal{w^2-3})^2 = 144 - 48w^2 + \gideal{w^2-3} = 144 - 48(3) + \gideal{w^2-3}= 0+ \gideal{w^2-3}$.
\end{example}

We can also view \cref{prob:solvePolyMap} as defining $|X|$ c-finite sequences of $A$ elements, where each sequence is the trajectory of a particular variable. If we take the special case where $A = \mathbb{Q}$ then the goal is to find the \emph{algebraic relations} \cite{KAUERS2008} of the $|X|$ sequences. 
\begin{definition}\label{def:algRelation}(Modification of \citet{KAUERS2008})
    Let $k$ be a field and $K$ a (commutative) $k$-algebra. An \textbf{algebraic relation over $k$} among $a_1, \dots, a_m \in K$ is an element of the kernel of the $k$-algebra homomorphism $\varphi:k[x_1, \dots, x_m]\rightarrow K$ that maps $x_j$ to $a_j$.
\end{definition}
\citet[Algorithm 2]{KAUERS2008} presents a method to find the set of algebraic relations over $\mathbb{Q}$ for the case of a set of c-finite sequences over the $\mathbb{Q}$-algebra, $\mathbb{Q}^\omega$; consequently, solving \cref{prob:solvePolyMap} for the case where $A = \mathbb{Q}$. In this section we show how the method of \citet{KAUERS2008} can be utilized to solve \cref{prob:solvePolyMap} for the case of an arbitrary $\mathbb{Q}$-algebra $A$. First we briefly review the method of \citet{KAUERS2008}.

At a high-level, the method of \citet{KAUERS2008} is, given c-finite sequences $\seq{a_1(n)}{n}, \dots, \seq{a_k(n)}{n} \in \mathbb{Q}^{\omega}$, perform the following:
\begin{enumerate}
    \item Compute closed-form solutions of each sequence as $a_i(n) = \sum_{j=1}^m p_{ij}(n) \Theta_j^n$ for polynomials $p_{ij} \in \mathbb{\bar{Q}}[n]$ and values $\Theta_1, \dots, \Theta_m\in \mathbb{\bar{Q}}$.
    \item Using the algorithm of \citet{Thesis:Ge} compute a basis $J\subseteq \mathbb{Q}[y_0, y_1, \dots, y_m]$ for the ideal of algebraic relations over $\mathbb{Q}$ of the sequences $\seq{n}{n}, \seq{\Theta_1^n}{n}, \dots, \seq{\Theta_m^n}{n} \in \mathbb{\bar{Q}}^\omega$. That is, $p(y_0, y_1, \dots, y_m) \in \gideal{J}$ if and only if $p(n, \Theta_1^n, \dots, \Theta_m^n) = 0$ for $n\in \mathbb{N}$.
    \item Let $B = \set{x_i - \sum_{j=1}^m p_{ij}(y_0)y_j : 1\le i \le k}$. Using Gr\"obner basis elimination techniques compute the ideal $\gideal{H} = (\gideal{B\cup J})\cap \mathbb{Q}[x_1, \dots, x_k]$.
\end{enumerate}
The resulting ideal, $\gideal{H}$, is the ideal of algebraic relations over $\mathbb{Q}$ of the sequences $\seq{a_1(n)}{n}$, $\dots$, $\seq{a_k(n)}{n}\in \mathbb{Q}^\omega$. That is, $\gideal{H}$ has the property that $p\in \gideal{H}$ if and only if $p\in \mathbb{Q}[x_1,\dots, x_k]$ and $p(a_1(n), \dots, a_k(n)) = 0$ for all $n\in \mathbb{N}$.
\begin{remark}
    It should be noted that the method presented above as well as in \citet{KAUERS2008}
only works when the characteristic polynomials of the input recurrences do not have 0 as a root. \citet{KAUERS2008} notes this, and correctly states such a situation can be handled with a pre-processing step. In this paper, we are more explicit on how to handle 0 roots.
\end{remark}

\begin{algorithm}[t]
\caption{Solve solvable map}\label{alg:solveMap}
	\SetAlgoLined
	\KwIn{$\mathbb{Q}$-algebra $A$, valuation $v\in A^X$, solvable map $f:\mathbb{Q}[X]\rightarrow \mathbb{Q}[X]$, and basis $I$ for the ideal $\I{A}{\set{v}}$.}
	\KwOut{A basis for the ideal $\I{A}{\set{f^n_A(v):i\in \mathbb{N}}}$}
    $m\gets |X|$\;
	For each $x_j \in X$ solve the trajectory of $x_j$ as $f_A^n(x_j) = 
 \sum_{i=1}^m \left( z_{ij}(n) + \sum_{k = 1}^m p_{ijk}(n) \Theta_k^n\right)v(x_i)$\;\label{Li:sequences}
    For $1\le i\le m$ and $1 \le j \le m$ let $a_{ij}(n) = \sum_{k=1}^mp_{ijk}(n)\Theta_k^n \in \mathbb{Q}^\omega$\;
    Using \citet[Algorithm 2]{KAUERS2008} let $J'\subseteq\mathbb{Q}[y_{11}, \dots, y_{1m}, \dots, y_{m1}, \dots, y_{mm}]$ be the ideal of algebraic relations of $\seq{a_{11}(n+m)}{n}, \dots, \seq{a_{1m}(n+m)}{n}, \dots, \seq{a_{m1}(n+m)}{n}, \dots, \seq{a_{mm}(n+m)}{n}$\;
    $J\gets$ Basis for the ideal $J' \cap \bigcap_{n=0}^{m-1} \left\langle \set{y_{ij} - (z_{ij}(n) + \sum_{k = 1}^m p_{ijk}(n) \Theta_k^n) : 1\le i\le m, 1\le j\le m}\right\rangle$\;
    \Return{$\textit{inv.image}(f,I\cup J)$, where $f$ is the map that sends $x_j \mapsto x_1y_{1j} + \dots + x_my_{mj}$}
\end{algorithm}
The main observation that leads to our method for a general $\mathbb{Q}$-algebra $A$ is that the only ``new'' algebraic relations over $\mathbb{Q}$ for the $A$ sequences must come from the relations of the initial values of the sequence (the ideal $\gideal{I} = \I{A}{\set{v}}$ in the statement of  \cref{prob:solvePolyMap}). This observation leads to our method for solving \cref{prob:solvePolyMap}, which we present as \cref{alg:solveMap}. 

\Cref{alg:solveMap} begins by writing the trajectory of each variable $x_j$ as the closed-form solution $\sum_{i=1}^m \left( z_{ij}(n) + \sum_{k = 1}^m p_{ijk}(n) \Theta_k^n\right)v(x_i)$, which is a sum-of-products of the form $a'_{1j}(n)v(x_1) + \dots + a'_{mj}(n)v(x_m)$ with each $\seq{a'_{ij}(n)}{n}\in \mathbb{Q}^\omega$. Moreover, each $a'_{ij}$ is a rational c-finite sequence. However, in general, the sequences $\seq{a'_{ij}(n)}{n}$ might have 0 as a root of their characteristic polynomials---hence the presence of the term $z_{ij}(n)$. These $z_{ij}(n)$ terms mean that the method of \cite{KAUERS2008} cannot be directly applied. The following lemma shows how we can handle this general case.
\begin{lemmarep}\label{lem:shiftRels}
    Let $\seq{a'_1(n)}{n}, \dots, \seq{a'_m(n)}{n}\in \mathbb{Q}^\omega$ such that $a'_i(n) = z_i(n) + a_i(n)$ for $\seq{z_i(n)}{n}$, $\seq{a_i(n)}{n}\in \mathbb{Q}^\omega$. Furthermore, suppose that there exists some $d\in \mathbb{N}$ such that for all $1\le i \le m$, $z_i(n) = 0$ for $n\ge d$. Let $J'\subseteq\mathbb{Q}[y_1,\dots, y_m]$ be the ideal of algebraic relations over $\mathbb{Q}$ of $\seq{a_1(n+d)}{n}$, $\dots$, $\seq{a_m(n+d)}{n}$. Then $J = J' \cap \bigcap_{n=0}^{d-1}\gideal{\set{y_{i} - (z_{i}(n) +a_i(n)): 1\le i\le m}}$ is the ideal of algebraic relations of $ \seq{a'_1(n)}{n}, \dots, \seq{a'_m(n)}{n}$.
\end{lemmarep}
\begin{appendixproof}
We need to show $p\in J$ if and only if $p(a'_1(n), \dots, a'_m(n)) = 0$ for $n\in \mathbb{N}$.
    Let $J'' = \bigcap_{n=0}^{d-1}\gideal{\set{y_{i} - (z_{i}(n) +a_i(n)): 1\le i\le m}}$.
    
    ($\implies$) Let $p\in J$. Then $p\in J'$ and $p\in J''$. Because $p\in J''$ it must be the case that $p = \sum_{i=1}^m g_{in}(y_i - (z_{i}(n) +a_i(n)))$ for $0\le n < d-1$ and some polynomials $g_{in}\in \mathbb{Q}[Y]$. Therefore, $p(a'_1(n), \dots, a'_m(n)) =\sum_{i=1}^m g_{in}(a'_i(n) - (z_{i}(n) +a_i(n))) = 0$ for $0\le n < d$. For $n\ge d$ we have $a'_i(n) = a_i(n)$. Equivalently $a'_i(n+d) = a_i(n+d)$ for $n\ge 0$. Because $p\in J'$, $p(a_1(n+d), \dots, a_m(n+d)) = 0$ for any $n\in \mathbb{N}$. Therefore, $p(a'_1(n+d), \dots, a'_m(n+d)) = 0$ for $n\in \mathbb{N}$ and $p(a'_1(n), \dots, a'_m(n)) = 0$ for $n\ge d$. Thus, $p(a'_1(n), \dots, a'_m(n)) = 0$ for $n\in \mathbb{N}$.

    ($\impliedby$) Suppose $p(a'_1(n), \dots, a'_m(n)) = 0$ for $n\in \mathbb{N}$. Then $p(a'_1(n), \dots, a'_m(n)) = 0$ for $0\le n < d$ and $n\ge d$. Thus, $p(a'_1(n+d), \dots, a'_m(n+d)) = p(a_1(n+d), \dots, a_m(n+d)) = 0$ for $n\in\mathbb{N}$, so $p \in J'$. For $0\le n < d$, $a'_i(n) = z_i(n) + a_i(n)$, so $p(a'_1(n), \dots, a'_m(n)) = p(z_1(n) + a_1(n), \dots, z_m(n) + a_m(n)) = 0$. Let $p(y_1, \dots, y_m)$ be reduced relative to a Gr\"obner basis for 
    \[\gideal{\set{y_{i} - (z_{i}(n) +a_i(n)): 1\le i\le m}}\]
    for some $0\le n < d$. Then $p = \sum_{i=1}^mg_{in}(y_i - (z_i(n) + a_i(n))) + r$ for some $g_{in}$'s $\in \mathbb{Q}[Y]$ and $r\in \mathbb{Q}$. Thus, we have 
    \[0 = p(z_1(n) + a_1(n), \dots, z_m(n) + a_m(n)) = \sum_{i=1}^mg_{in}(z_i(n) + a_i(n) - (z_i(n) + a_i(n))) + r = r.\]
    Thus, $p\in \gideal{\set{y_{i} - (z_{i}(n) +a_i(n)): 1\le i\le m}}$ for $0\le n < d$. Therefore, $p\in J''$.
\end{appendixproof}

Because each $\seq{a'_{ij}(n)}{n}$ is c-finite, each $z_{ij}(n)$ has the property that $z_{ij}(n) = 0$ for $n\ge d$, where $d$ is the maximal order of the sequences $\seq{a'_{ij}(n)}{n}$. In the case of \cref{alg:solveMap} we have $d \le m$. Thus, each $z_{ij}(n) = 0$ for $n\ge m$.
\Cref{lem:shiftRels} shows that $J$ in \cref{alg:solveMap} is a basis for the ideal of algebraic relations over $\mathbb{Q}$ of the $\seq{a'_{ij}(n)}{n}$ sequences.

We use the next two lemmas to establish the needed property of the return value of \cref{alg:solveMap} and thus show that \cref{alg:solveMap} is correct. The desired result of \cref{alg:solveMap} is the ideal of algebraic relations over $\mathbb{Q}$ of the $A$ sequences, $f^n_A(x_j)$ defined in Line~\ref{Li:sequences} of \cref{alg:solveMap}. These sequences are defined as a sum-of-products of rational sequences and constant sequences of the valuations $v(x_i)$. The next lemma (\cref{lem:algebraic-relations}) shows that if we want to find the algebraic relations over $\mathbb{Q}$ among arbitrary rational sequences lifted to $A^\omega$ and constant valuations $\set{\seq{v(x_j)}{i}: x_j\in X}$, it is sufficient to consider the ideals of algebraic relations over $\mathbb{Q}$ among $\set{\seq{v(x_j)}{i}: x_j\in X}$ and the lifted rational sequences separately. This is what makes \cref{alg:solveMap} possible: we are given the algebraic relations over $\mathbb{Q}$ among $\set{\seq{v(x_j)}{i}: x_j\in X}$ as input, and we can calculate the algebraic relations over $\mathbb{Q}$ among the c-finite sequences defined in Line~\ref{Li:sequences} using the algorithm of \citet{KAUERS2008}. The second lemma (\cref{lem:solvableRelations}) then applies \cref{lem:algebraic-relations} to the specific form of the $\seq{f^n_A(x_j)}{n}$ sequences defined in \cref{alg:solveMap} to establish the correctness of the algorithm.
\begin{lemma} \label{lem:algebraic-relations}
  Let $\seq{a_{1}(i)}{i}, \dots,\seq{a_{m}(i)}{i} \in \mathbb{Q}^\omega$ and let $J\subseteq \mathbb{Q}[y_{1}, \dots, y_{m}]$ be the ideal of algebraic relations over $\mathbb{Q}$ among $\seq{a_{1}(i)}{i}, \dots,\seq{a_{m}(i)}{i}$. Let $A$ be a $\mathbb{Q}$-algebra with additive unit $0_A$ and multiplicative unit $1_A$. Let $X = \set{x_1,\dots, x_n}$. Let $v\in A^X$ and let $I = \I{A}{\set{v}}\subseteq \mathbb{Q}[X]$.
  Let $\varphi : \mathbb{Q}[X, Y] \rightarrow A^\omega$ be the $\mathbb{Q}$-algebra homomorphism defined by
    \begin{align*}
        \varphi(x_j) = \seq{v(x_j)}{i} \qquad
        \varphi(y_j) = \seq{a_j(i)(1_A)}{i}
    \end{align*}
  Then $\varphi(p) = \seq{0_A}{i}$ if and only if $p\in \gideal{I\cup J}\subseteq \mathbb{Q}[X, Y]$.
\end{lemma}
\begin{proof}  
  ($\impliedby$) Suppose $p\in \gideal{I\cup J}$. Thus, $p=g + h$ for some $g\in \gideal{I}\subseteq \mathbb{Q}[X,Y]$ and $h\in \gideal{J} \subseteq \mathbb{Q}[X,Y]$.

  Because $g\in \gideal{I}$, $g = \sum_{j=1}^{k_1}f_jg_j$ for $f_j\in \mathbb{Q}[X, Y]$ and $g_j\in I$. Because $g_j\in I = \I{A}{\set{v}}$, $g^A_j(v) = 0_A$ for each $j$. Therefore, $\varphi(g) = \sum_j^{k_1} \varphi(f_j)\seq{g_j^A(v)}{i} = \sum_j^{k_1} \varphi(f_j)\seq{0_A}{i} = \seq{0_A}{i}$.

  Now we show the same for $h$. Because $h\in \gideal{J}$, $h = \sum_{j=1}^{k_2} f_jh_j$ for $f_j\in\mathbb{Q}[X, Y]$ and $h_j\in J\subseteq \mathbb{Q}[Y]$. Because $h_j\in J$, $\varphi(h_j) = \seq{h_j(a_1(i), \dots, a_m(i))(1_A)}{i} = \seq{0}{i}$ for each $j$. Therefore, $\varphi(h) = \sum_{j=1}^{k_2} \varphi(f_j)\varphi(h_j) = \sum_{j=1}^{k_2} \varphi(f_j)\seq{0_A}{i} = \seq{0_A}{i}$.
  
  Combining the previous paragraphs we have that, because $\varphi$ is a homomorphism, $\varphi(p) = \varphi(g) + \varphi(h) = \seq{0_A}{i} + \seq{0_A}{i} = \seq{0_A}{i}$.
  
  ($\implies$) Suppose $\varphi(p) = \seq{0_A}{i}$. 
  Let $r$ be $p$ reduced by a Gr\"obner basis for $\gideal{I\cup J}$ under the elimination order $\ll_X$. That is $p = g_I + g_J + r$ for some $g_I\in I$, $g_J\in J$ and for any other $g'\in \gideal{I\cup J}$ and $r'\in \mathbb{Q}[X,Y]$ with $p = g'+r'$, $\LM(r') \gg_X \LM(r)$. 
  Because $\varphi$ is a homomorphism we have $\seq{0_A}{i}=\varphi(p) = \varphi(g_I+g_J+r) = \varphi(g_I)+\varphi(g_J)+\varphi(r) = \seq{0_A}{i} +\varphi(r)=\varphi(r)$. 
  We can write $r$ by collecting $X$ monomials with respect to the order $\gg_X$ as follows
  \begin{align*}
      r = &m^X_1p^Y_1(y_1, \dots, y_m) + \dots + m^X_kp^Y_k(y_1,\dots, y_m) + \\
      &p^Y_{k+1}(y_1,\dots, y_m)
  \end{align*}
  where each $m^X_s$ is a distinct monomial of $X$ variables with $m^X_1 \gg_X m^X_j$ for $j = 2, \dots, k$, and each $p^Y_s$ is a polynomial in $\mathbb{Q}[Y]$.

  Observe that there exists some $i$ such that $p_1^Y(a_1(i), \dots, a_m(i))\neq 0$. If not, $\varphi(p_1^Y(y_1, \dots, y_m)) = \seq{p_1^Y(a_1(i),\dots,a_m(i))(1_A)}{i} = \seq{0_A}{i}$. Thus, $p_1^Y(a_1(i),\dots,a_m(i)) = 0$ for $i\in \mathbb{N}$, and therefore $p_1^Y$ is an algebraic relation over $\mathbb{Q}$ among $\seq{a_1(i)}{i},\dots,\seq{a_m(i)}{i}$. Then by definition $p_1^Y\in J$. But this contradicts the property that $r$ is reduced with respect to $J$. That is, if $p_1^Y(y_1, \dots, y_m) \in \gideal{J}$, then there exists a better $r'$ that does not contain the monomial $m_1^X$. Therefore, there is some $i$ such that $p_1^Y(a_1(i), \dots, a_m(i))\neq 0$.

  Let $i$ be such that $p_1^Y(a_1(i), \dots, a_m(i))\neq 0$. We have $\varphi(r)=\seq{0_A}{i}$ by assumption, and so $\varphi(r)_i = 0_A$, where
  \begin{equation*}
  \begin{split}
      \varphi(r)_i = &(m^X_1)^A(v)p^Y_1(a_1(i), \dots, a_m(i)) + \dots + (m^X_k)^A(v)p^Y_k(a_1(i), \dots, a_m(i)) + \\
      &p^Y_{k+1}(a_1(i), \dots, a_m(i))\\
  \end{split}
  \end{equation*}
 denotes the $i$th element of $\varphi(r)$.
  Because $\varphi(r)_i = 0_A$, we must have
  \begin{equation}\label{Eq:ReducePoly}
  \begin{split}
      &m^X_1p^Y_1(a_1(i), \dots, a_m(i)) + \dots + m^X_kp^Y_k(a_1(i), \dots, a_m(i)) + \\
      &p^Y_{k+1}(a_1(i), \dots, a_m(i)) \in I\\
  \end{split}
  \end{equation}
  Denote the polynomial in \labelcref{Eq:ReducePoly} as $h$. We can rewrite $r$ as follows
  \begin{align*}
      r &= \frac{p_1^Y(y_1,\dots, y_m)}{p_1^Y(a_1(i), \dots, a_m(i))} h + r'
  \end{align*}
  For some $r'$ containing $x$ monomials $m^X_2, \dots, m^X_k$. However, this (nearly) contradicts the property that $p$ is reduced with respect to $I$. That is, $p = (g_I + g_J +\frac{p_1^Y(y_1,\dots, y_m)}{p_1^Y(a_1(i), \dots, a_m(i))} h) + r'$ with $g_I + g_J+\frac{p_1^Y(y_1,\dots, y_m)}{p_1^Y(a_1(i), \dots, a_m(i))} h \in \gideal{I \cup J}$. Moreover, $\LM(r) \gg_X \LM(r')$, because $r'$ does not contain the monomial $m_1^X$ and $m_1^X\gg_X m_j^X$ for $j=2, \dots, k$. The only way to avoid the contradiction is to have $m^X_1$ be a constant. Therefore, because $r$ is reduced with respect to an order that eliminates $X$ variables and $\LM(r)\in \mathbb{Q}[Y]$, we have $r\in \mathbb{Q}[Y]$.

  Finally, because $r\in \mathbb{Q}[Y]$ and $\varphi(r) = \seq{0_A}{i}$, $r$ must be an algebraic relation over $\mathbb{Q}$ among $\seq{a_1(i)}{i}, \dots, \seq{a_m(i)}{i}$, and therefore $r \in J$. However, because $r$ must be reduced with respect to $J$, $r = 0$. Thus, $p = g_I +g_J$ with $g_I+g_J\in I+J$. Therefore, $p\in \gideal{I \cup J}$.
\end{proof}

\begin{lemmarep}\label{lem:solvableRelations}
  Let $\seq{a_{1m}(i)}{i}, \dots,\seq{a_{1m}(i)}{i}, \dots, \seq{a_{m1}(i)}{i}, \dots, \seq{a_{mm}(i)}{i}\in \mathbb{Q}^\omega$ and let $J\subseteq \mathbb{Q}[y_{11}, \dots, y_{1m}, \dots, y_{m1}, \dots, y_{mm}]$ be the ideal of algebraic relations over $\mathbb{Q}$ among these sequences. Let $A$ be a $\mathbb{Q}$-algebra, $X = \set{x_1,\dots, x_m}$, $v\in A^X$, and $I = \I{A}{\set{v}}\subseteq \mathbb{Q}[X]$.
    For every $i\in \mathbb{N}$ let $w^i \in A^{X}$ be defined as $w^i(x_j) = a_{1j}(i)v(x_1) + \dots +a_{mj}(i)v(x_m)$. Let $f : \mathbb{Q}[X] \rightarrow \mathbb{Q}[X,Y]$ be the polynomial homomorphism that maps $x_j \mapsto x_1y_{1j} + \dots + x_my_{mj}$ for all $j$. Then $\I{A}{\set{w^i : i \in \mathbb{N}}} = f^{-1}[\gideal{I \cup J}]$.
\end{lemmarep}
\begin{appendixproof}
Let $\varphi : \mathbb{Q}[X,Y] \rightarrow A^\omega$ be defined as in \cref{lem:algebraic-relations}.
Let $p(x_1, \dots, x_m) \in \mathbb{Q}[X]$. Then $p(x_1y_{11} + \dots+x_my_{1m}, \dots, x_1y_{m1} + \dots+x_my_{mm})\in \mathbb{Q}[X,Y]$. Note that 
\begin{align*}
    &\varphi(p(x_1y_{11} + \dots+x_my_{1m}, \dots, x_1y_{m1} + \dots+x_my_{mm})) \\
    &= p(\varphi(y_{11})\varphi(x_1) + \dots + \varphi(y_{1m})\varphi(x_m),\dots, \varphi(y_{m1})\varphi(x_1) + \dots + \varphi(y_{mm})\varphi(x_m))\\
        &= \seq{p(a_{11}(i)v(x_1) + \dots + a_{1m}(i)v(x_m), \dots, a_{m1}(i)v(x_1) + \dots + a_{mm}(i)v(x_m))}{i} \\
        &=\seq{p^A(w^i)}{i}.
    \end{align*}
    Therefore, by \cref{lem:algebraic-relations}, $\seq{p^A(w^i)}{i} = 0$ if and only if $p(x_1y_{11} + \dots+x_my_{1m}, \dots, x_1y_{m1} + \dots+x_my_{mm})\in \gideal{I\cup J}$.
    So we have the following chain
    \begin{align*}
        &p \in \I{A}{\set{w^i} : i\in \mathbb{N}}\\
        &\iff \seq{p^A(w^i)}{i} = 0\\
        &\iff p(x_1y_{11} + \dots+x_my_{1m}, \dots, x_1y_{m1} + \dots+x_my_{mm})\in \gideal{I\cup J}\\
        &\iff p(x_1, \dots, x_m)\in f^{-1}(\gideal{I \cup J})\\
    \end{align*}
\end{appendixproof}

Combining \cref{lem:inverse-image,lem:shiftRels,lem:solvableRelations} establishes the correctness of \cref{alg:solveMap}.
\begin{theorem}
    \Cref{alg:solveMap} solves \cref{prob:solvePolyMap}.
\end{theorem}

\subsection{Closing Transition Ideals}
The result of \cref{alg:solveMap} produces a polynomial ideal that summarizes the algebraic relations over $\mathbb{Q}$ of a solvable polynomial map. In this subsection, we show how the result of \cref{alg:solveMap} can be used to compute $T^*$ for a (ultimately) solvable transition ideal $T$. Recall that every solvable transition ideal $T$ comes with a solvability witness $p$. The basic idea to compute $T^*$ is to use \cref{alg:solveMap} to summarize the algebraic relations over $\mathbb{Q}$ among the sequences defined by $p$. However, \cref{prob:solvePolyMap}, which is solved by \cref{alg:solveMap}, is defined over a $\mathbb{Q}$-algebra $A$. Hence, in this subsection we work to motivate and explain that in order to calculate $T^*$, the $\mathbb{Q}$-algebra we want to instantiate \cref{prob:solvePolyMap} with is $A =\mathbb{Q}[X, X']/\dom^*(T)$.

Before we talk of $\mathbb{Q}$-algebras, we make a brief observation of the structure of solvable transition ideal. Intuitively, a solvable transition ideal can be broken into a domain part, containing only unprimed variables, and a transition part. The next lemma formalizes this point.
\begin{lemmarep}\label{lem:solvableIdealSyn}
    Let $T\subseteq \mathbb{Q}[X,X']$ be a solvable transition ideal, and let $p : \mathbb{Q}[X] \rightarrow \mathbb{Q}[X]$ be a solvable witness for $T$. Then $T^n = \dom(T^n) + \gideal{\set{x_i' - p^n(x_i): 1\le i\le m}}$ for $1 \le n$.
\end{lemmarep}
\begin{appendixproof}
    We prove the lemma by induction on $n$.
    Let $n=1$. Because $T$ is solvable, $x_i' - p(x_i)\in T$ for $x_i\in X$. Thus, a Gr\"obner basis for $T$ with respect to $\elimorder{X'}$ is of the form
    \[
    \gideal{f_1(X), \dots, f_k(X), x_1'-p(x_1), \dots, x'_m - p(x_m)} = \dom(T) + \gideal{\set{x_i' - p(x_i): 1\le i\le m}}.
    \]

    Now suppose the lemma holds for $n$. We wish to show the lemma holds for $T^{n+1}$.
    \begin{align*}
        T^{n+1} = &T \cdot T^n\\
        = &(T[X'\rightarrow X''] + T^n[X\rightarrow X'']) \cap \mathbb{Q}[X, X']\\
        = &(\dom(T)[X'\rightarrow X''] + \gideal{\set{x_i'' - p(x_i): 1\le i\le m}} \\
        &+ \dom(T^n)[X\rightarrow X''] + \gideal{\set{x_i' - p^n(x_i)'': 1\le i\le m}}) \cap \mathbb{Q}[X, X']\\
        = &(\dom(T) + \dom(T^n)[X\rightarrow p(X)] + \\
        &\gideal{\set{x_i' - p^n(p(x_i)): 1\le i\le m}} + \gideal{ \set{x_i'' - p(x_i): 1\le i\le m}}) \cap \mathbb{Q}[X, X']\\
        = &\dom(T^{n+1}) + \gideal{\set{x_i' - p^{n+1}(x_i): 1\le i\le m}}.
    \end{align*}
\end{appendixproof}

From \cref{lem:solvableIdealSyn}, we see that the iterated behavior of the transition ideal is \emph{mostly} captured by the iterated behavior of the polynomial witness; what is missing is the $\dom(T^n)$ part. Note that $\dom(T^n)$ is an ideal for each $n$. If we let $A$ be the $\mathbb{Q}$-algebra $\mathbb{Q}[X,X']/I$ for some ideal $I$, we can define a sequence like the one in \cref{prob:solvePolyMap} that uses a solvable witness $p$ to transition not only variables but sets of polynomials with respect to $I$. This can be formalized in the language of \cref{prob:solvePolyMap} for a solvable transition ideal $T\subseteq\mathbb{Q}[X,X']$ with solvability witness $p:\mathbb{Q}[X] \rightarrow \mathbb{Q}[X]$ as follows:
\begin{itemize}
    \item Let $A = \mathbb{Q}[X,X']/I$.
    \item Let $\hat{p}:\mathbb{Q}[X,X']\rightarrow\mathbb{Q}[X,X']$ extend $p$ as $\hat{p}(x_i) = x_i$ if $x_i\in X$, and $\hat{p}(x_i') = (p(x_i))'$ if $x_i'\in X'$.
    \item Let $v\in A^{X\cup X'}$ be defined as $v(x_i) = x_i + I$ if $x_i\in X$ and $v(x_i') = p(x_i) + I$ if $x_i'\in X'$.
\end{itemize}
The question then is what should we take for $I$. We want $I$ to be the polynomials not captured by the iteration of the solvable witness. Thus, from \cref{lem:solvableIdealSyn} these are the polynomials in $\dom(T^n)$. For \cref{prob:solvePolyMap} $I$ needs to be fixed, so we have two obvious options for $I$, the domain of $T$, $\dom(T)$, or the invariant domain of $T$, $\dom^*(T)$. The next example shows what happens if we use the domain of $T$ and why that gives us a mismatch for computing $T^*$.
\begin{example}\label{ex:notInvDom}
    Consider the solvable transition ideal $T = \gideal{y' - y - z, z' - 3z, y-z-1}$, with $\dom(T) = \gideal{y - z - 1}=I$. A solvable witness for $T$ is $p(y) = y+z$ and $p(z) = 3z$. Let $A = \mathbb{Q}[X, X']/I$, $\hat{p}:\mathbb{Q}[X,X']\rightarrow\mathbb{Q}[X,X']$ extend $p$ as above, and let $v\in A^{X\cup X'}$ with $v(y) = y+I$, $v(z) = z+I$, $v(y') = y+z + I$, and $v(z') = 3z+I$. In the language of \cref{prob:solvePolyMap}, we have $\I{A}{\set{v}} = T$. $\hat{p}_A$ defines the following sequence\footnote{Note that $y+z+\gideal{y-z-1} =  2z+1+\gideal{y-z-1}$}:
    \[
    \left(\begin{Bmatrix}
        y\mapsto y+I\\
        z \mapsto z+I\\
        y' \mapsto 2z + 1 + I\\
        z' \mapsto 3z + I
    \end{Bmatrix}, 
    \begin{Bmatrix}
        y\mapsto y+I\\
        z \mapsto z+I\\
        y' \mapsto 5z + 1 + I\\
        z' \mapsto 9z + I
    \end{Bmatrix}, 
    \begin{Bmatrix}
        y\mapsto y+I\\
        z \mapsto z+I\\
        y' \mapsto 14z + 1 + I\\
        z' \mapsto 27z + I
    \end{Bmatrix}, 
    \dots
    \right).
    \]
    Taking $\I{A}{\set{v^i}}$ for each $i$ of the above sequence is nearly $T, T^2, T^3, \dots$. However, $T^2 = \langle y' - y - z, z' - 3z, y - 1, z\rangle = \gideal{y' - 1, z', y - 1, z}$, but $y' -1 \not\in \I{A}{\set{v^2}}$. Moreover, for \emph{every} $i\ge 2$, $\I{A}{\set{v^i}} \neq T^i$.
\end{example}
The essential problem with the previous example is that the domain of $T$ is not stable for higher iterations of $T^i$. If instead of $I = \dom(T)$ in the previous example we used $I= \dom^*(T)$ then we would have the equality $\I{A}{\set{v^i}} = T^i$ for $i\ge 2$. This observation that equality can be recovered for some $i$ by using the \emph{invariant domain} is our key insight for computing $T^*$. 
The issue of \cref{ex:notInvDom} is fixed using the reasoning in the following lemma.
\begin{lemmarep}\label{lem:transitionIdealclosure}
    Let $T\subseteq \mathbb{Q}[X,X']$ be a solvable transition ideal with solvability witness $p$. Let $N\ge 1$ be such that $\dom(T^N) = \dom^*(T)$, and define $I \defeq \dom^*(T)$. Let $A$ be the $\mathbb{Q}$-algebra $\mathbb{Q}[X,X']/I$. Let $\hat{p} : \mathbb{Q}[X, X'] \rightarrow \mathbb{Q}[X, X']$ be the homomorphism defined by $\hat{p}(x_i) = x_i$ and $\hat{p}(x'_i) = p(x_i)'$. Let $v\in A^{X\cup X'}$ be the valuation defined by $v(x_i) = x_i + I$ and $v(x'_i) = p(x_i) + I$. Then for $1\le N \le n$, $T^{n} = \I{A}{\set{\hat{p}_A^{n-1}(v)}}$. Furthermore, if $1\le n < N$ then $T^{n} \subseteq \I{A}{\set{\hat{p}_A^{n-1}(v)}}$
\end{lemmarep}
\begin{appendixproof}
   Consider $\hat{p}^{n-1}_A(v)$ for $n\ge 1$. On a variable $x_i \in X$ we have $\hat{p}^{n-1}_A(v)(x_i) = v(x_i) = x_i + I$. For a variable $x'_i\in X'$ we have 
   \[
   \hat{p}^{n-1}_A(v)(x_i') = (\hat{p}^{n-1}(x_i'))^A(v) = p(p^{n-1}(x_i)) + I = p^n(x_i) + I.
   \]

   More succinctly,
   \[
\hat{p}^{n-1}_A(v) = \begin{Bmatrix}
        x_i\mapsto x_i + I\\
        x'_i \mapsto p^{n}(x_i) + I
    \end{Bmatrix}.
\]
Note that if $n\ge N$, $\dom(T^n) = \dom(T^N) = \dom^*(T) = I$. But if $n < N$, $\dom(T^n)\subseteq \dom^*(T) = I$. Thus, if $n\ge N$ we have
\begin{align*}
    \I{A}{\set{\hat{p}_A^{n-1}(v)}} = &\gideal{\set{x_i - q(x): q(x) \in x_i + I, 1\le i \le m}} +\\ 
    &\gideal{\set{x'_i - q(x): q(x) \in p^n(x_i) + I, 1\le i \le m}}\\
    = & I + \gideal{\set{x'_i - p^n(x_i) : 1 \le i \le m}}\\
    = & \dom(T^n) + \gideal{\set{x'_i - p^n(x_i) : 1 \le i \le m}}\\
    = & T^n
\end{align*}
The justification for the last step comes from \cref{lem:solvableIdealSyn}. If $n < N$ then the second to last equals becomes $\supseteq$.
\end{appendixproof}
\begin{corollaryrep}
    \Cref{lem:transitionIdealclosure} also holds for ultimately solvable transition ideals.
\end{corollaryrep}
\begin{appendixproof}
    Let $T$ be an ultimately solvable transition ideal. Then by definition $T + \dom^*(T)$ is solvable. By \cref{lem:transitionIdealclosure} we have $(T+\dom^*(T))^n\subseteq  \I{A}{\set{\hat{p}_A^{n-1}(v)}}$ for $1\le n < N$, and $(T+\dom^*(T))^n=\I{A}{\set{\hat{p}_A^{n-1}(v)}}$ for $n\ge N$. It can readily be verified that $(T+\dom^*(T))^n = T^n + \dom^*(T)$ for any transition ideal $T$. Furthermore, if $n\ge N$ then $T^n+\dom^*(T) = T^n$. Therefore, for $1\le n < N$ 
    \[T^n \subseteq T^n + \dom^*(T) = (T+\dom^*(T))^n \subseteq \I{A}{\set{\hat{p}_A^{n-1}(v)}}.\]
    For $n\ge N$, 
    \[T^n = T^n + \dom^*(T) = (T+\dom^*(T))^n = \I{A}{\set{\hat{p}_A^{n-1}(v)}}.\]
\end{appendixproof}

\begin{example}\label{ex:withInvDom}
    Recall \cref{ex:notInvDom}, but with $I = \dom(T^2) = \dom(T^3) = \dom^*(T) = \gideal {y-1, z}$. $\hat{p}_A$ defines the following sequence:
    \[
    \left(\begin{Bmatrix}
        y\mapsto 1+I\\
        z \mapsto 0+I\\
        y' \mapsto 1 + I\\
        z' \mapsto 0 + I
    \end{Bmatrix}, 
    \begin{Bmatrix}
        y\mapsto 1+I\\
        z \mapsto 0+I\\
        y' \mapsto 1 + I\\
        z' \mapsto 0 + I
    \end{Bmatrix},
    \begin{Bmatrix}
        y\mapsto 1+I\\
        z \mapsto 0+I\\
        y' \mapsto 1 + I\\
        z' \mapsto 0 + I
    \end{Bmatrix},
    \dots
    \right).
    \]
\end{example}
Informally, \cref{lem:transitionIdealclosure} states that the long-running relations of a solvable transition ideal $T$ is \emph{exactly} captured by $\I{A}{\set{\hat{p}^i(v) : i\in \mathbb{N}}}$. This is what \cref{ex:withInvDom} shows. However, we cannot just take the long-running relations of $T$ as our summary $T^*$. This is because for iterations before the invariant domain has stabilized we do not have equality. However, the invariant domain must stabilize in a finite number of iterations, and can then be recovered via ideal intersection. This leads to the following theorem.
\begin{theoremrep}\label{the:tstar}
    Let $T$ be a (ultimately) solvable transition ideal $T\subseteq \mathbb{Q}[X,X']$ with solvability witness $p$. Let $N\ge 1$ be such that $\dom(T^N) = \dom^*(T)$ = I.  Let $\hat{p} : \mathbb{Q}[X, X'] \rightarrow \mathbb{Q}[X, X']$ be the homomorphism defined by $\hat{p}(x_i) = x_i$ and $\hat{p}(x'_i) = p(x_i)'$. Let $v\in A^{X\cup X'}$ be the valuation defined by $v(x_i) = x_i + I$ and $v(x'_i) = p(x_i) + I$. Then
    \[
    \bigcap_{i=0}^\infty T^i = \left(\bigcap_{i=0}^{N-1} T^i\right) \cap (\I{A}{\set{\hat{p}^i(v) : i\in \mathbb{N}}}).
    \]
\end{theoremrep}
\begin{appendixproof}
\begin{align*}
    \left(\bigcap_{i=0}^{N-1} T^i\right) \cap \left(\I{A}{\set{\hat{p}^i(v) : i\in \mathbb{N}}}\right) &= \left(\bigcap_{i=0}^{N-1} T^i\right) \cap \left(\bigcap_{i=1}^\infty \I{A}{\set{\hat{p}^{i-1}(v)}}\right)\\
    &= \left(\left(\bigcap_{i=0}^{N-1} T^i\right) \cap \left(\bigcap_{i=1}^{N-1} \I{A}{\set{\hat{p}^{i-1}(v)}}\right) \cap \left(\bigcap_{i=N}^{\infty} \I{A}{\set{\hat{p}^{i-1}(v)}}\right)\right)\\
    &= T^0 \cap \left(\bigcap_{i=1}^{N-1} (T^i \cap (\I{A}{\set{\hat{p}^{i-1}(v)}}))\right) \cap \left(\bigcap_{i=N}^{\infty} \I{A}{\set{\hat{p}^{i-1}(v)}})\right)\\
    &= T^0 \cap \left(\bigcap_{i=1}^{N-1} T^i\right) \cap \left(\bigcap_{i=N}^{\infty} T^i\right) \\
    &= \bigcap_{i=0}^\infty T^i
\end{align*}
The second to last step is justified by \cref{lem:transitionIdealclosure}.
\end{appendixproof}
The right-hand-side of the equation in \cref{the:tstar} is computable. The term $\I{A}{\set{\hat{p}^i(v) : i\in \mathbb{N}}}$ can be computed by \cref{alg:solveMap} with $\I{A}{\set{v}} = \dom^*(T) + T$. The term $(\bigcap_{i=0}^{N-1} T^i)$ is a finite intersection of polynomial ideals which can be computed via Gr\"obner basis techniques. Asymptotically, the Gr\"obner basis calculations for computing intersections as well as the Gr\"obner basis calculations in \cref{alg:solveMap} dominate the running time, making the overall computation exponential.

\begin{example}\label{ex:tstar}
    Recall $T = \gideal{y' - y - z, z' - 3z, y-z-1}$ from \cref{ex:notInvDom}.
    \[
    T^* = \gideal{3z^2-4zz' + (z')^2, yz'-yz + z^2-zz' - z' + z, 2y' - 2y - z' + z}
    \]
\end{example}


\section{Loop summarization modulo LIRR}
\label{sec:summary}
Loop summarization is the problem of computing, for a given transition formula $F$ representing the body of some loop, an over-approximation of the reflexive transitive closure of $F$.
This section describes how to combine the components introduced in the previous two sections to accomplish this task.  We prove the key property that our loop summarization procedure is monotone.  Finally, we discuss how this procedure can be combined with other summarization techniques to enhance the ability of an algebraic program analyzer to generate non-linear loop summaries.

Our iteration operator takes a four-step approach (pictured \cref{fig:overviewDiag}).  Given an input transition formula $F$,
\begin{enumerate}
\item Compute the transition ideal $\I{\LIRR}{F}$ of $F$ (using the algorithm of \citet{POPL:KKZ2023})
\item Compute a solvable reflection $\tuple{t,T}$ of $\I{\LIRR}{F}$ (\cref{sec:reflection})
\item Compute $T^*$ (\cref{sec:closure})
\item Calculate the formula corresponding to the image of $T^*$ under $t$.
\end{enumerate}
More succinctly, we define an operator $\fstar{(-)} : \TF \rightarrow \TF$ to be
\[ \fstar{F} \defeq \iformula{\image{t}{\istar{T}}} \]
where $\tuple{t,T}$ is a solvable reflection of $\I{\LIRR}{F}$.
Naturally, one may repeat this recipe for defining a loop summarization operator by using ultimately solvable transition ideals (\cref{sec:ultimately-solvable}) and/or polynomial simulations (\cref{sec:polynomial-simulations}), and the soundness and monotonicity results that we prove below hold also for these variants.

\begin{example}\label{Ex:idealFromForm}
Consider the transition formula $F$ below, and its associated transition ideal:
\[F = \left(\begin{array}{r@{}l}
& (k' = k + 1) \land (x' = y) \land (y' = x)\\
\land & \left(\begin{array}{r@{}l}
& (z \geq 0 \land z' = w + z \land w' = x^2)\\
\lor & (\lnot (z \geq 0) \land w' = w + z \land z' = x^2)
\end{array}
\right)
\end{array} \right)
\hspace{0.25cm}
\I{\LIRR}{F} =
 \gideal{
 \begin{array}{l}
w'x^2 - w' z' + x^2 z' - x^2,\\
-w' + w + x^2 - z' + z,\\
x' - y,\\
y' - x,\\
k' - k - 1
\end{array}
}
\]
Notice that, while $F$ employs a rich logical language involving disjunction, negation, and inequalities, its ideal $\I{\LIRR}{F}$ is defined
by the set of polynomials $p$ such that $F$ entails $p = 0$.
A solvable reflection of $\I{\LIRR}{F}$ is
$\tuple{t,T}$ where $t$ is the map that sends $a \mapsto x$, $b \mapsto y$, $c \mapsto (w+z)$, and $d \mapsto k$, and $T$ is the ideal
$\gideal{ a' - b, b' - a, c' - c - a^2, d' - d - 1}$.  The closure of $T$
is $T^* = \gideal{ab - bb' + (b')^2 - ab', a'+b' - a - b, b^2d'+a^2d'-a^2d-b^2d-b^2+ab'+bb'-ab - 2c' + 2c}$. The (ideal generated by) the image of
$T^*$ under $\dvext{t}$ is 
\[\gideal{
\begin{array}{c}
      xy - yy' + (y')^2 - xy',x'+y' - x - y,\\
      y^2k'+x^2k'-x^2k-y^2k-y^2+xy'+yy'-xy - 2(w' + z') + 2(w+z)
\end{array}}.\]
Finally, we have $\fstar{F} \defeq xy + (y')^2 = yy' + xy' \land x'+y' = x+y \land y^2k'+x^2k'+xy'+yy' - 2(w' + z') = x^2k + y^2k+y^2+xy - 2(w+z)$.
\end{example}

\begin{theoremrep}[Soundness]
Let $F$ be a transition formula.  For any $n \in \mathbb{N}$, we have
$F^n \models_{\LIRR} \fstar{F}$
\end{theoremrep}
\begin{appendixproof}
    Let $\tuple{t,T}$ be a solvable reflection of $\I{\LIRR}{F}$.
    We may show that
\[
    \I{\LIRR}{F^n} \supseteq \I{\LIRR}{F}^n \supseteq \image{t}{T^n} \supseteq \image{t}{T^*}
    \]
    by induction on $n$.  The base case $n=0$ is trivial.  The induction step follows from the fact that (1) the sequential composition operator for ideals over-approximates the sequential composition operator for transition formulas, and (2) sequential composition for transition ideals preserves simulation.
\end{appendixproof}

\begin{theorem}[Monotonicity]
  Let $F$ and $G$ be transition formulas. 
  If $F \models_\LIRR G$, then $\fstar{F} \models_\LIRR \fstar{G}$.
\end{theorem}
\begin{proof}
    Suppose $F \models_{\LIRR} G$.  Let $\tuple{t,T}$ be a solvable reflection of $\I{\LIRR}{F}$, and let $\tuple{u,U}$ be a solvable reflection of $\I{\LIRR}{G}$.
    Since $F \models_{\LIRR} G$, we must have $\I{\LIRR}{F} \supseteq \I{\LIRR}{G}$, and thus $u$ is a simulation from $\I{\LIRR}{F}$ to $U$.  Since $U$ is solvable and $\tuple{t,T}$ is a solvable reflection of $\I{\LIRR}{F}$, there is a (unique) simulation $v : T \rightarrow U$ such that $u = \simcompose{v}{t} = t \circ v$.  Since $v$ is a simulation, we have
    $\istar{T} \supseteq \image{v}{\istar{U}}$, and so
    \[ \image{t}{\istar{T}} \supseteq \image{t}{\image{v}{\istar{U}}} = \image{(t \circ v)}{\istar{U}} = \image{u}{\istar{U}}  \]
    Since $\fstar{F} = \iformula{\image{t}{\istar{T}}}$
    and $\fstar{G} = \iformula{\image{u}{\istar{U}}}$, we have
    $\fstar{F} \models_\LIRR \fstar{G}$.
\end{proof}

\subsection{Modular Design of Loop Summarization Operators}

The loop summarization operator that is defined in this paper is designed to compute polynomial invariants.  Such invariants are often just \textit{a} component of a correctness argument for a program---for example, a correctness argument may rely upon reasoning about inequalities, or may require a disjunctive invariant.  Our loop summarization operator can be incorporated in a broader invariant generation scheme by using various combinators to combine summarization operators.  For instance, the simplest such combinator is a \textit{product}, which combines two loop summarization $\oast_1$ and $\oast_2$ into one $\oast_1 \times \oast_2$ by taking their conjunction:
\[ F^{\oast_1 \times \oast_2} \defeq F^{\oast_1} \land F^{\oast_2}\]
Provided that both the $\oast_1$ and $\oast_2$ operators are monotone, then (1) so is their product, and (2) the resulting analysis is at least as precise as either component analysis.

Another kind of summarization combinator is the refinement technique proposed in \citet{POPL:CBKR2019}.  This combinator exposes phase structure in loops, and in particular enables a "base" summarization operator that may only generate conjunctive invariants to produce disjunctive invariants.  In addition to monotonicity, this combinator requires four additional axioms in order to guarantee that it improves analysis precision.  The following proposition states that indeed our summarization operator satisfies these conditions.

\begin{propositionrep}
Let $F$ be a transition formula.   Then the following hold
\begin{itemize}
\item (Reflexivity) $1 \models_{\LIRR} \fstar{F}$
\item (Extensivity) $F \models_{\LIRR} \fstar{F}$
\item (Transitivity) $\fstar{F} \circ \fstar{F} \equiv_{\LIRR} \fstar{F}$
\item (Unrolling) For any natural number $n$, $(F^n)^* \models_{\LIRR} F^*$.
\end{itemize}
\end{propositionrep}
\begin{appendixproof}
    Reflexivity, Exensivity, and Transitivity are straightforward. We shall prove unrolling.
    
    Let $n$ be a natural number.
  Let $\tuple{t,T}$ be a solvable reflection of $\I{\LIRR}{F^n}$, and let $\tuple{u,U}$ be a solvable reflection of $\I{\LIRR}{F}$.  Since
  $\I{\LIRR}{F}^n \subseteq \I{\LIRR}{F^n}$, $\tuple{t,T}$ is a solvable reflection of $\I{\LIRR}{F^n}$, and $U^n$ is solvable, there is a (unique) simulation $v : U^n \rightarrow T$ such that $t \circ v = u$.
  Since $\fstar{(F^n)} = \iformula{\image{t}{\istar{T}}}$
  and $\fstar{(F^n)} = \iformula{\image{u}{\istar{U}}}$, it is sufficient to prove that
  $\image{t}{\istar{T}} \supseteq \image{u}{\istar{U}}$.
  Observe that since $v : U^n \rightarrow T$ is a simulation, we have
  \[
  \image{t}{\istar{T}} = \image{t}{\bigcap_{i=0}^\infty T^i}
  \supseteq
\image{t}{\bigcap_{i=0}^\infty \image{v}{U^{ni}}}
  \supseteq
\image{t}{\image{v}{\bigcap_{i=0}^\infty U^{ni}}}
  \]
  Then since $t \circ v = u$, we have
  \[
\image{t}{\istar{T}}
\supseteq \image{t}{\image{v}{\bigcap_{i=0}^\infty U^{ni}}}
= \image{(t \circ v)}{\bigcap_{i=0}^\infty U^{ni}}
= \image{u}{\bigcap_{i=0}^\infty U^{ni}}
\supseteq \image{u}{\bigcap_{i=0}^\infty U^{i}} = \image{u}{U^*} \qedhere
  \]
\end{appendixproof}

\section{Experimental evaluation} \label{sec:evaluation}
We consider two experimental questions concerning our methods for synthesizing loop invariants for general programs:
\begin{enumerate}
    \item (\Cref{sec:exp-ques-gen}) How do our techniques apply to the task of verifying general programs?
    \item (\Cref{sec:exp-ques-aligator}) How do our techniques for generating polynomial invariants perform on programs for which other tools guarantee completeness? 
\end{enumerate}
In relation to each of these questions we also want to understand the performance, both in terms of accuracy and running time, of using linear simulations as well as polynomial simulations of bounded degree for extracting solvable transition ideals from transition ideals. 
\subsection{Experimental Setup}
    \paragraph{Implementation} We implemented the techniques described in this paper in a tool called \Tool. Our implementation relies on 
    \begin{itemize}
        \item Chilon and ChilonInv \cite{POPL:KKZ2023}, for \LIRR~operations and generating invariant inequalities, respectively.
        \item The \textsc{FGb} library \cite{fgb} for an implementation of the F4 algorithm \cite{FAUGERE199961}, which we use for computing of Gr\"obner bases.
        \item \textsc{Flint} \cite{flint} for integer lattice computations and \textsc{Arb} \cite{arb} for numerical polynomial root finding. These operations are required to implement the algorithm of \citet{Thesis:Ge} used in \cref{alg:solveMap}.
    \end{itemize}

    \Tool can be configured to use either linear or quadratic simulations, and
    either solvable or ultimately solvable transition ideals.  Our testing
    revealed that (1) the difference between using solvable and ultimately
    solvable is negligible (both in success rate and runtime
    performance), and (2) the cost of na\"{i}ve computation of the full
    inverse image $\invimage{f_{X,2}}{-}$ for quadratic simulations is
    prohibitively high.  In the following, we report on two configurations of
    \Tool: \textit{USP-Lin} is the product of the ChilonInv domain and iteration
    operator induced by ultimately solvable linear reflections, \textit{USP-Quad} is
    the product of USP-Lin and the iteration operator induced by solvable
    quadratic simulations with a single stratum (which necessitates only
    computing the affine polynomials in $\invimage{f_{X,2}}{-}$, and is
    therefore more tractable).

\paragraph{Environment} We ran all experiments on a virtual machine (using Oracle VirtualBox), with a guest OS of Ubuntu 22.04 allocated with 8 GB of RAM, using a 4-core Intel Core i7-4790K CPU @ 4.00 GHz. All tools were run with the \textsc{BenchExec} \cite{bechexec} tool using a time limit of 300 seconds on all benchmarks.

\paragraph{Benchmarks} Our 202 benchmarks programs are sourced from the set of safe\footnote{That is, the error location is truly unreachable.} benchmarks from the \texttt{c/ReachSafety-Loops} subcategory of the Software Verification Competition (SV-COMP) \cite{svcomp}. We divided our 202 benchmarks into a loops category consisting of 176 programs, and an NLA category consisting of 26 benchmarks. The NLA benchmarks are modified versions of the programs in the \texttt{nla-digbench} set from SV-COMP, intended to evaluate the strength of \Tool's ability to generate non-linear invariants.  The \texttt{nla-benchmark} programs from SV-COMP have ``proposed invariants'' at each loop header, as well as assertions at the end of the programs as post conditions; we obtained the NLA suite by removing these ``proposed invariants''.  As a result, non-linear invariants must be \textit{synthesized} in order to prove the post-condition (rather than simply \textit{verifying} that the proposed invariant is an invariant, and implies the post-condition).  The program from \cref{fig:overviewEx} is an example of a program in the NLA suite.


\paragraph{Comparison Tools} We have compared our techniques with ChilonInv \cite{POPL:KKZ2023}, CRA \cite{POPL:KCBR2018}, VeriAbs 1.5.1-2 \cite{veriabs}, and ULTIMATE Automizer 0.2.3 \cite{uautomizer}. ChilonInv and CRA use a similar verification strategy of extracting implied solvable invariants of loop bodies to generate invariants of loops. VeriAbs and ULTIMATE Automizer are high performers at SV-COMP and provide context to the overall results. The strategies of ChilonInv, USP-Lin, and USP-Quad are all monotone algebraic analyses and the refinement technique of \citet{POPL:CBKR2019} applies. Refinement is guaranteed to improve the precision of these three techniques, and so we have employed refinement in the comparison of these three strategies.

\subsection{How do our Techniques Perform on a Suite of General Verification Tasks?}\label{sec:exp-ques-gen}

\begin{table}[t]
\caption{Comparison of tools on the loops and NLA benchmarks. T represents the amount of time, in seconds, take by each tool not including timeouts nor out of memory exceptions. The number of timeouts is reported in parentheses. We also experienced out of memory exceptions with VeriAbs which are noted in parentheses. \#P represents the number of benchmarks proved correct. The best results in each category is bolded.\label{tab:results}}

{\small
\begin{tabular}{clrrr}
\multicolumn{1}{l}{}                                                          &                     & loops             & NLA                                   & Total             \\ \hline
\multicolumn{1}{l}{}                                                          & \multicolumn{1}{l|}{\#B} & 176               & \multicolumn{1}{r|}{26}               & 202               \\ \hline
\multirow{2}{*}{ChilonInv}                                                    & \multicolumn{1}{l|}{\#P} & 144               & \multicolumn{1}{r|}{1}                & 145               \\
                                                                              & \multicolumn{1}{l|}{T}   & 974 (5) & \multicolumn{1}{r|}{\textbf{38 (1)}} & 1010 (6) \\ \hline
\multirow{2}{*}{USP-Lin}                                                      & \multicolumn{1}{l|}{\#P} & 149               & \multicolumn{1}{r|}{8}                & 157               \\
                                                                              & \multicolumn{1}{l|}{T}   & 1130 (5)          & \multicolumn{1}{r|}{97 (1)}          & 1230 (6)          \\ \hline
\multirow{2}{*}{USP-Quad}                                                     & \multicolumn{1}{l|}{\#P} & 122               & \multicolumn{1}{r|}{\textbf{15}}      & 132               \\
                                                                              & \multicolumn{1}{l|}{T}   & 1750 (31)          & \multicolumn{1}{r|}{93 (6)}          & 1840 (37)         \\ \hline
\multirow{2}{*}{CRA}                                                          & \multicolumn{1}{l|}{\#P} & \textbf{154}      & \multicolumn{1}{r|}{8}                & \textbf{162}      \\
                                                                              & \multicolumn{1}{l|}{T}   & \textbf{669 (5)}           & \multicolumn{1}{r|}{133 (8)}          & \textbf{802 (13)}         \\ \hline
\multirow{2}{*}{VeriAbs}                                                      & \multicolumn{1}{l|}{\#P} & 116               & \multicolumn{1}{r|}{2}                & 118               \\
                                                                              & \multicolumn{1}{l|}{T}   & 3430 (55, 2 OOM)         & \multicolumn{1}{r|}{42 (23, 1 OOM)}        & 3470 (78, 3 OOM)         \\ \hline
\multirow{2}{*}{\begin{tabular}[c]{@{}c@{}}ULTIMATE\\ Automizer\end{tabular}} & \multicolumn{1}{l|}{\#P} & 125               & \multicolumn{1}{r|}{9}                & 134               \\
                                                                              & \multicolumn{1}{l|}{T}   & 2270 (51)         & \multicolumn{1}{r|}{247 (17)}        & 2520 (68)        
\end{tabular}
}
\end{table}

\Cref{tab:results} gives the results of running each tool on the program verification benchmarks. Theoretically, in terms of precision, ChilonInv $\preceq$ USP-Lin $\preceq$ USP-Quad. However, this does not consider timeouts. Due to the increased power of USP-Lin and USP-Quad we would expect in terms of time taken ChilonInv $\preceq$ USP-Lin $\preceq$ USP-Quad, and this is what we see reflected in \cref{tab:results}. In our experiments we found USP-Lin to outperform ChilonInv in both the loops category and the NLA category in terms of programs verified, at the expensive of additional running time. Theoretically, USP-Quad is stronger than ChilonInv and USP-Lin; however, the extra power comes at a price of running time. As can be seen from \cref{tab:results} USP-Quad performed \emph{worse} on the loops category compared with ChilonInv and USP-Lin because of the number of timeouts. However, due to its strong non-linear reasoning capability, USP-Quad outperformed all the other tools on the difficult NLA benchmarks.

Theoretically, USP-Lin and USP-Quad are incomparable with the other tools. On one hand CRA's recurrence extraction procedure is weaker than the methods in this paper. However, CRA is also able to produce invariants involving exponential and polynomial terms, whereas the techniques in this paper are only able to produce invariants involving polynomial terms. VeriAbs is a portfolio of many different techniques, such as bounded model checking and k-induction. ULTIMATE Automizer implements a trace abstraction algorithm. We note that while USP-Lin outperformed VeriAbs and ULTIMATE Automizer on the loops category, VeriAbs and ULTIMATE Automizer have additional capabilities such as the ability to produce counterexamples in the case when an assertion does not hold. This capability is outside the scope of USP-Lin and USP-Quad. Nevertheless, we find USP-Lin to be quite competitive on our benchmark suite. It outperformed all other tools on the loops category except for CRA, where it is behind by only 5 examples. Moreover, because of the success of USP-Quad on the NLA suite we find that powerful techniques that generate polynomial invariants are required to verify interesting programs found in the literature.

\subsection{How do our Techniques Compare with Prior Methods for Complete Generation of Polynomial Invariants?}\label{sec:exp-ques-aligator}
\begin{wraptable}{r}{0.4\textwidth}
    \caption{USP-Lin and USP-Quad on the multi-path \textsc{Aligator} benchmarks.}
    \begin{tabular}{l|c|c}
        & USP-Lin & USP-Quad \\
        \hline
        egcd.c & \xmark & \cmark\\
        fermat2.c & \xmark & \cmark\\
        lcm2.c & \xmark & \cmark\\
        divbin.c & \xmark & \xmark\\
        prodbin.c & \xmark & \cmark\\
        dijkstra.c & \xmark & \xmark
    \end{tabular}
    \label{tab:aligatorbench}
\end{wraptable}
In this subsection, we consider how our method for generating polynomial invariants (which works on general programs) compares with the method presented by \citet{VMCAI:HJK2018} (which is complete, but applies to a more limited class of programs). The method of \citet{VMCAI:HJK2018} is implemented in a tool called \textsc{Aligator}. Both our methods of linear simulations as well as polynomial simulations are complete for loops whose bodies are described by a solvable polynomial map. \textsc{Aligator} is also complete for such loops. However, the completeness result of \citet{VMCAI:HJK2018} also extends to multi-path loops, where each branch is described by a solvable polynomial map (e.g., a loop of the form \texttt{while(*)\{ if (*) A else B \}}, where \texttt{A} and \texttt{B} are described by solvable polynomial maps). On such an example, \textsc{Aligator} will produce \emph{all} polynomial invariants of the loop, but \Tool cannot make the same guarantee. At the level of a loop we \emph{abstract} the loop body to a solvable transition ideal. In the case of \texttt{while(*)\{ if (*) A else B \}}, we create a solvable transition ideal that abstracts \emph{both} \texttt{A} and \texttt{B}, which is strictly weaker than considering \texttt{A} and \texttt{B} separately as in \citet{VMCAI:HJK2018}.

We investigated how USP-Lin and USP-Quad perform on multi-path loops for which \textsc{Aligator} is complete, but our techniques are incomplete. 
A direct practical comparison between USP-Lin, USP-Quad, and \textsc{Aligator} is challenging because they take different formats as input. However, a subset of 6 programs from the multi-path benchmark suite of \textsc{Aligator} are applicable for our tool\footnote{\textsc{Aligator} and \Tool treat integer division differently}. All of these 6 programs are found in the NLA suite discussed in \cref{sec:exp-ques-gen}.
More detailed results of running USP-Lin and USP-Quad on these six examples can be found in \cref{tab:aligatorbench}. The completeness result of \citet{VMCAI:HJK2018} applies to these 6 programs, so given enough time \textsc{Aligator} would be able to verify all 6 of them. As can be seen from \cref{tab:aligatorbench} USP-Lin, is unable to verify any of the 6 programs; however, USP-Quad is able to verify 4 of the 6. 
For the other 2 programs, the reason USP-Quad is unable to succeed is because those examples perform integer division in a situation in which no round-off occurs. In these programs this property is essentially encoded with an exponential invariant, which is outside the capabilities of USP-Quad.
From the results of \cref{tab:aligatorbench} we conclude that while the class of loops for which our technique is complete is a subset of \textsc{Aligator}'s, we can still generate most of the invariants needed to prove correctness.


\section{Related work} \label{sec:related}

\paragraph{Polynomial abstractions of loops}
The algorithm in \cref{sec:reflection} for computing the solvable reflection of a transition ideal can be seen as both a
refinement of \citet{POPL:KCBR2018}'s algorithm for extracting a solvable polynomial map from a transition formula
and a generalization of \citet{CAV:ZK2021}'s algorithm for computing
deterministic affine reflections.
Contrasting with \cite{POPL:KCBR2018}, our algorithm is guaranteed to find a \textit{best} abstraction as a solvable transition ideal,
which is essential to prove monotonicity of our analysis.
Contrasting with \cite{CAV:ZK2021}, our algorithm consumes and produces transition ideals, which generalize affine relations.

\citet{SAS:ABKKMS2022} considers the problem of abstracting polynomial endomorphisms by solvable polynomial maps.  The technique presented in \cref{sec:reflection} is more general in the sense that it operates on \textit{transition ideals} rather than polynomial endomorphisms.  A polynomial endomorphism $p : \mathbb{Q}[X] \rightarrow \mathbb{Q}[X]$ can be encoded as a transition ideal,
generated by the polynomials $\set{x' - p(x) : x \in X}$, in which case the algorithm in \cref{sec:reflection-computation} computes a solvable transition ideal (from which we may recover a solvable polynomial map---that is, our procedure serves the same purpose as of \citet{SAS:ABKKMS2022} for the inputs considered in that work).  Moreover, our procedure provides a precision guarantee: it finds \textit{solvable reflections} of transition ideals.   

For example, consider the loop below (left) along with its solvable reflection (right)
\[
\textbf{while} * \textbf{do } \left\{\begin{array}{l}
x := x + z^2 + 1;\\
y := y - z^2;\\
z := z + (x+y)^2
\end{array}\right\}\hspace{1cm}
\underbrace{\tuple{
  \set{ a \mapsto x+y, b \mapsto z},
  \gideal{a' - a - 1, b' - b - a}
  }}_{\text{Solvable reflection}}
\]
While the technique in \cite{SAS:ABKKMS2022} is able to identify the first
polynomial in the reflection (corresponding to the update $(x+y) := (x+y) +
1$) it cannot find the second ($z' := z + (x+y)^2$), since there is a
non-linear dependence of $z$ upon the ``defective'' variables $x$ and $y$
whose dynamics cannot be described by a solvable polynomial map.


\citet{SAS:FHG2021} considers another related problem: \textit{given a polynomial endomorphism $p$, is there a polynomial automorphism $f$ such that $f^{-1} \circ p \circ f$ is solvable?} The procedure in \cref{sec:reflection} can also be used to solve this problem: if $\tuple{t,T}$ is the solvable reflection of $p$, then such an $f$ exists (namely, $t$) exactly when the ambient dimension of $T$ is equal to that of $p$ (and $T$ has real eigenvalues). \Cref{sec:reflection} generalizes this result in the sense that, (1) we operate on transition ideals rather than polynomial endomorphisms and (2) should the answer to the decision problem be ``no'', we may still compute an \textit{abstraction} of $p$.

\paragraph{Complete polynomial invariant generation} \citet{TACAS:Kovacs2008,VMCAI:HJK2018,JACM:EOPW2023,LICS:HOPW2018,ISAAC:RCK2004} are complete methods for generating polynomial invariants on limited program structures. Our method matches the completeness results of these works on single loops whose bodies are described by solvable polynomial maps; however, the completeness result of each of these works cover additional situations. 

\citet{LICS:HOPW2018,JACM:EOPW2023} present a method that is complete for generating polynomial invariants for affine programs where all branching represents non-deterministic choice. Our methods have no issue analyzing such programs. Moreover, our method can also reason about programs with polynomial assignments as well as branching with conditionals. However, even though our method can reason about general affine programs, we can only guarantee completeness in the case of a loop whose body is described by a solvable polynomial map.

\citet{TACAS:Kovacs2008} presents complete polynomial invariant generation for \emph{P-solvable} loops. These are loops, with no branching, whose bodies have either Gosper-summable or c-finite assignments. As stated in \cref{sec:c-finite} c-finite sequences are equivalent to solvable polynomial maps, and so our technique matches \citet{TACAS:Kovacs2008} in that regard. However, while we always extract a solvable transition ideal from a loop, solvable transition ideals are not powerful enough to capture certain Gosper-summable examples. Thus, while our method is monotone on such examples, it does not guarantee completeness. \citet{VMCAI:HJK2018} extends \citet{TACAS:Kovacs2008} to the case of multi-path loops where each branch has a body with Gosper-summable or c-finite recurrence assignments. In the Gosper-summable case the comparison is the same as \citet{TACAS:Kovacs2008}. In the multi-path c-finite case we are also not complete; however, we experimentally compare with \citet{VMCAI:HJK2018} in \cref{sec:exp-ques-aligator}. In either the case of \citet{TACAS:Kovacs2008} or \citet{VMCAI:HJK2018} they cannot make a completeness guarantee for programs having branching with conditionals or programs with arbitrary loop nesting.

\citet{ISAAC:RCK2004,JSC:RCK2007} present a complete method for the case of a single multi-path loop where each branch has a body described by a c-finite recurrence. This matches the c-finite case of \citet{VMCAI:HJK2018}, except \citet{ISAAC:RCK2004,JSC:RCK2007} have an additional restriction on the eigenvalues of the c-finite recurrences (corresponding to the $\Theta_i$ variables of \cref{Eq:c-finiteClosedForm}). For \citet{ISAAC:RCK2004,JSC:RCK2007} the eigenvalues are required to be positive and rational. We have no such restriction and so our method generalizes \citet{ISAAC:RCK2004,JSC:RCK2007} in the case of a simple loop where the body is described by a c-finite recurrence. However, their completeness result goes beyond our capability in the case of a multi-path loop with positive rational eigenvalues.

\paragraph{Template Based Methods} Another method for generating polynomial invariants is to reduce the problem to constraint solving by supposing that the invariant takes the form of some parameterized template, and solving for the parameters \cite{KOJIMA201833,ATVA:OBP16,POPL:SSM2004,SAS:CJJK2012,MULLEROLM2004,PLDI:CFGG2020,GHMM2023}. These methods have the benefit of being able to handle problems with arbitrary control flow. Furthermore, they are often complete for generating invariants that fit the given template. Many template methods consider all polynomials up to some bounded degree. In such cases when the desired polynomial is within the degree bound, template based methods have the potential to generate invariants for general programs that our method would theoretically miss. In contrast, our method does not require a degree bound. Even for linear simulations, there is no bound on the degree of the invariant our method calculates.

\paragraph{Monotone algebraic program analysis}
A recent line of work has used the framework of algebraic program analysis to develop program analyses with monotonicity guarantees
\cite{CAV:SK2019,PLDI:ZK2021,CAV:ZK2021,POPL:KKZ2023}.
In particular, \citet{POPL:KKZ2023} proposes a monotone loop summarization algorithm based on the theory of linear integer/real rings.  Our technique is complementary in the sense that our method computes stronger invariant polynomial equations than \cite{POPL:KKZ2023}, but cannot synthesize invariant polynomial inequalities.


\section*{Data-Availability Statement}
An implementation of \Tool and experimental scripts are available on Zenodo \cite{artifact}. 
\iflong
\else
Omitted proofs can be found in \citet{POPL:CK2024}.
\fi

\begin{acks}
We would like to thank Tom Reps for his contributions to the discussions that lead to this paper.
This work was supported, in part,
by a gift from Rajiv and Ritu Batra, a Google PhD fellowship, and by the NSF under grant number 1942537.
Any opinions, findings, and conclusions or recommendations
expressed in this publication are those of the authors,
and do not necessarily reflect the views of the sponsoring
entities.
\end{acks}

\bibliographystyle{ACM-Reference-Format}
\bibliography{references}

\end{document}